\pdfoutput=1

\documentclass[ acmsmall ]{acmart}
\usepackage{mathtools}
\usepackage{mathpartir}
\usepackage{cleveref}
\usepackage{subcaption}
\usepackage{listings}
\usepackage{paralist}

\usepackage{tikz}
\usetikzlibrary{calc}
\usetikzlibrary{positioning}
\usetikzlibrary{arrows}

\newcommand{\command}[1]{\texttt{#1}}
\newcommand{\comAssert}[1]{\command{assert}~#1}
\newcommand{\comITE}[3]{\command{if}~#1~\command{then}~#2~\command{else}~#3}
\newcommand{\comIT}[2]{\command{if}~#1~\command{then}~#2}
\newcommand{\comWhile}[2]{\command{while}~#1~\command{do}~#2}
\newcommand{\comLabel}[1]{\command{label}~#1}
\newcommand{\comBrk}{\command{break}}
\newcommand{\comRet}{\command{return}}
\newcommand{\comGoto}[1]{\command{goto}~#1}
\newcommand{\true}{\command{true}}
\newcommand{\false}{\command{false}}
\newcommand{\conti}[1]{\mathbf{#1}}
\newcommand{\acc}[1]{{\conti{acc}~#1}}  % accept with an output indicator value
\newcommand{\ret}{{\conti{ret}}}  % return, terminate the program
\newcommand{\brk}[1]{\conti{brk}~#1}  % break with an output indicator value
\newcommand{\jmp}[1]{{\conti{jmp}~#1}} % goto a label

\newcommand{\At}{\mathbf{\mathrm{At}}}

\newcommand{\BExp}{\mathbf{\mathrm{BExp}}}
\newcommand{\Exp}{\mathbf{\mathrm{Exp}}}

\newcommand{\contWith}{\mathbf{cont}}
\newcommand{\exitWith}{\mathbf{exit}}
\newcommand{\iter}{\mathrm{iter}}

\crefname{prog}{program}{programs}
\Crefname{prog}{Program}{Programs}

\theoremstyle{acmdefinition}
\newtheorem*{remark}{Remark}

\setcopyright{rightsretained}
\acmDOI{10.1145/3704857}
\acmYear{2025}
\acmJournal{PACMPL}
\acmVolume{9}
\acmNumber{POPL}
\acmArticle{21}
\acmMonth{1}
\received{2024-07-11}
\received[accepted]{2024-11-07}

\begin{document}

\title{\texorpdfstring{CF-GKAT:\@ Efficient Validation of Control-Flow Transformations}{CF-GKAT: Efficient Validation of Control-Flow Transformations}} % chktex 13

\author{Cheng Zhang}
\orcid{0000-0002-8197-6181}
\affiliation{%
  \institution{University College London}
  \city{London}
  \country{United Kingdom}
}
\email{czhang03@bu.edu}
\authornote{Work performed at Boston University.}

\author{Tobias Kappé}
\orcid{0000-0002-6068-880X}
\affiliation{%
  \institution{LIACS, Leiden University}
  \city{Leiden}
  \country{Netherlands}
}
\email{t.w.j.kappe@liacs.leidenuniv.nl}
\authornote{Work performed at Open University of the Netherlands and ILLC, University of Amsterdam.}

\author{David E. Narváez}
\orcid{0000-0003-3704-1060}
\affiliation{%
  \institution{Virginia Tech}
  \city{Blacksburg}
  \country{USA}
}
\email{david.narvaez@vt.edu}

\author{Nico Naus}
\orcid{0000-0003-3442-1543}
\affiliation{%
  \institution{Open University of the Netherlands}
  \city{Heerlen}
  \country{Netherlands}
}
\affiliation{%
  \institution{Virginia Tech}
  \city{Blacksburg}
  \country{USA}
}
\email{nico.naus@ou.nl}


\begin{abstract}
 % !TEX root=../paper.tex
%
%The abstract should briefly summarize the contents of the paper in
% To all: I've included my favorite abstract template below, I always find it very useful. - Nico

% Context
Guarded Kleene Algebra with Tests (GKAT) provides a sound and complete framework to reason about trace equivalence between simple imperative programs.
% Inquiry
However, there are still several notable limitations.
First, GKAT is completely agnostic with respect to the meaning of primitives, to keep equivalence decidable.
Second, GKAT excludes non-local control flow such as \command{goto}, \command{break}, and \command{return}.
% Approach
To overcome these limitations, we introduce \emph{Control-Flow GKAT} (\emph{CF-GKAT}), a system that allows reasoning about programs that include non-local control flow as well as hardcoded values.
% Knowledge
CF-GKAT is able to soundly and completely verify trace equivalence of a larger class of programs, while preserving the nearly-linear efficiency of GKAT\@.
This makes CF-GKAT suitable for the verification of control-flow manipulating procedures, such as decompilation and \command{goto}-elimination.
% Grounding
To demonstrate CF-GKAT's abilities, we validated the output of several highly non-trivial program transformations, such as Erosa and Hendren's \command{goto}-elimination procedure and the output of Ghidra decompiler.
% Importance
CF-GKAT opens up the application of Kleene Algebra to a wider set of challenges, and provides an important verification tool that can be applied to the field of decompilation and control-flow transformation.

% Context: What is the broad context of the work? What is the importance of the general research area?
% Inquiry: What problem or question does the paper address? How has this problem or question been addressed by others (if at all)?
% Approach: What was done that unveiled new knowledge?
% Knowledge: What new facts were uncovered? If the research was not results-oriented, what new capabilities are enabled by the work?
% Grounding: What argument, feasibility proof, artifacts, or results and evaluation support this work?
% Importance: Why does this work matter?

\end{abstract}





\maketitle

%%%%%%%%%%%%%%%%%%%% NOTATION NOTE  %%%%%%%%%%%%%%%%%%%%
%
%%%%% FONTS %%%%%
% a model of (co)algebra: uppercase mathcal  (\mathcal{G})
% a theory mathsf  (\mathsf{KAT})
% a command: mathtt (\mathtt{while})
% a continuation: mathbf (\mathbf{jmp})
%
%%%%% VARIABLES %%%%%
% a primitive test: t
% set of all primitive test: T
% a boolean expression: b, c, d
% a primitive action: p, q, r
% set of all primitive actions:  \Sigma
% an expressions: e, f
% a state: s
% state set: S
% pseudo-state: h
% a guarded string: w
% a guarded language: G
% guarded language model: \mathcal{G}
% a continuation: c
% set of all continuations: C
% a guarded string followed by a continuation: w  \cdot  c
% a guarded language with continuation: G  \times  C
% continuation model (guarded language): \mathcal{C}
% a label:  \ell
% set of all labels: L
% indicator value: i, j
% indicator variable: x
% set of all indicator value: I

\section{Introduction}

There are many notions of program equivalence with different granularity and computational complexity.
At one end of the spectrum, syntactic equality compares two programs solely based on its syntax. Although decidable, this technique will not equate programs like the following:
\begin{mathpar}
 \comITE{t}{p}{q}
 \and
 \comITE{ \neg  t}{q}{p}
\end{mathpar}
Conversely, input-output equivalence relates two programs if and only if they yield the same output on the same input.
This equivalence is very powerful, but also well-known to be undecidable.

Situated between these two extremes, \emph{Guarded Kleene Algebra with Tests}~\cite{kozen_BohmJacopiniTheorem_2008a,smolka_GuardedKleeneAlgebra_2020},
or \emph{GKAT} for short, reasons about \emph{trace equivalence} between simple \command{while}-programs.
By abstracting the meaning of the primitive tests and actions, it can focus on how tests determine which actions are performed.
For example, the two programs above are equivalent in GKAT:\@ they both execute $p$ when $t$ holds, and $q$ when it does not.
GKAT is also able to validate nontrivial equivalences like
\[
  \comWhile{t}{p};\ \comWhile{s}{\{\ q;\ \comWhile{t}{p}\ \}}  \equiv
 \comWhile{t  \lor  s}{\{\ \comITE{t}{p}{q}\ \}}
\]
Surprisingly, GKAT equivalence is decidable in nearly-linear time (assuming the set of test variables is fixed)~\cite{smolka_GuardedKleeneAlgebra_2020}, making it a reasonable compromise between complexity and granularity.

Nevertheless, the abstraction of GKAT comes at the price of not being able to validate some straightforward equivalences.
First, as mentioned, GKAT disregards the meaning of primitive programs and tests.
For instance, when given a program like
\begin{equation}
 \comITE{y \neq 0}{\{\ x := 42;\ p\ \}}{\{\ x := 42;\ q\ \}},%
 \label[prog]{prog: assignment inside branches}
\end{equation}
we can note that a change in the value of $x$ does not have effect on whether $y \neq 0$.
Hence, it should be possible to factor the assignment to $x$ out of the branches, and obtain
\begin{equation}
 x := 42;\ \comITE{y \neq 0}{p}{q}.%
 \label[prog]{prog: assignment outside branches}
\end{equation}
GKAT does not admit this equivalence, precisely because it is agnostic with respect to the meaning of primitive actions.
However, moving to a setting that accounts for the semantics of actions is hard, because Turing completeness --- and by extension, undecidability --- lurks nearby~\cite{kozen_KleeneAlgebraTests_1996, kuznetsov_ComplexityReasoningKleene_2023,azevedodeamorim_KleeneAlgebraCommutativity_2024}

Second, GKAT excludes non-local control-flow constructs like \(\command{goto}\), \(\comBrk\), and \(\comRet\).
In a general imperative language, this does not limit expressivity, as these constructs can be recovered using variables~\cite{erosa-hendren-1994} --- indeed, the Böhm-Jacopini theorem says that every type of control flow can be written using a single $\texttt{while}$-loop~\cite{DBLP:journals/cacm/BohmJ66}.
However, lacking both variables and non-local control structures, GKAT is not able to express all control flow in real-world programs.

As a concrete example, consider the programs below.
The control flow in \Cref{prog: goto version two state automaton} is based purely on labels and $\texttt{goto}$.
Meanwhile, \Cref{prog: break version two state automaton} is structured as a loop with the option to terminate early using $\comBrk$.
These programs happen to be trace equivalent (i.e., they always execute the same actions in the same order) but represent behavior not expressible in plain GKAT~\cite{schmid_GuardedKleeneAlgebra_2021}.
\begin{align}
  & \begin{aligned}
      & \comLabel{ \ell _{0}};\ \comIT{\neg t}{\comGoto{ \ell _{1}}};\ p;\ \comIT{t}{\comGoto{ \ell _{1}}};\ q;\ \comGoto{ \ell _{0}};\ \comLabel{ \ell _{1}}
    \end{aligned}\label[prog]{prog: goto version two state automaton}
 \\[3pt]
  & \comWhile{t}{\{\ p;\ \comITE{ \neg  t}{q}{\comBrk}\ \}}
 \label[prog]{prog: break version two state automaton}
\end{align}

As it turns out, deciding equivalence between programs such as these is essential in validating control-flow manipulation procedures.
Specifically, consider the control flow structuring phase of a decompiler~\cite{cifuentes-1994}, which is tasked with converting conditional and unconditional jumps into more conventional control flow constructs.
\Cref{prog: goto version two state automaton} can be thought of as pseudo-assembly that models the input of this process, and \Cref{prog: break version two state automaton} is a plausible outcome of decompilation.
Thus, the control-flow structuring process is correct when \Cref{prog: goto version two state automaton,prog: break version two state automaton} are equivalent.

To overcome the limitations of GKAT, we propose control-flow GKAT (CF-GKAT), an extension that is capable of equating some interesting programs.
Concretely, we extend GKAT in two ways.

First, we add \emph{indicator variables}, which can be assigned and tested against hardcoded values, and do not appear in other primitive actions and tests.
For example, assignments like $x := 42$ are allowed, but assignments like $x := y + 1$ are not.
This strikes a delicate balance, by being weak enough to exclude general computation (thus keeping equivalence decidable), yet strong enough to model the equivalence of \Cref{prog: assignment inside branches,prog: assignment outside branches}.
The addition of indicator variables empowers CF-GKAT to validate control-flow transformation algorithms that use them~\cite{yakdan_NoMoreGotos_2015,erosa-hendren-1994}.
We focus on a single indicator variable for brevity; an extension to multiple indicator variables should be straightforward.

Second, we extend GKAT with the non-local control-flow constructs \(\comGoto{\!}\), \(\comBrk\) and \(\comRet\).
This poses a challenge, as the non-local nature of these commands prevents a compositional semantics --- after all, the precise meaning of a statement like \(\comBrk\) depends on its context.
To overcome this, we propose an intermediate \emph{continuation semantics}.
In this semantics, each trace is tagged with a ``continuation'' that can accept (terminate normally), break, return, or go to a label.
Then, the trace semantics of the program can be obtained by resolving these continuations.

%FIXME: CZ: This paragraph is still bit out of place to me
\smallskip
We were able to design an automaton model for CF-GKAT, where every CF-GKAT expression can be converted into a CF-GKAT automaton while preserving the continuation semantics.
Furthermore, CF-GKAT automata and continuation semantics can be lowered into GKAT automata and trace semantics respectively, while preserving their semantic correspondence.
We are thus able to reduce the problem of deciding trace equivalence of programs to deciding bisimilarity of two GKAT automata, which can be done efficiently.
Consequently, CF-GKAT can soundly and completely validate trace equivalence of a larger class of programs, while preserving the nearly-linear efficiency of GKAT\@.
For instance, it can automatically validate that \Cref{prog: break version two state automaton,prog: goto version two state automaton} are equivalent to each other, and also to their single-loop equivalent obtained via the Böhm-Jacopini theorem~\cite{DBLP:journals/cacm/BohmJ66}:
\begin{align}
 \begin{aligned}
   & x := 1;\ \comWhile{x \neq 0}{\{ \\
   & \qquad \comITE{x = 1  \land  t}{\{\ p;\ x := 2\ \} \\
   & \qquad\qquad}{\comITE{x = 2  \land   \neg  t} {\{\ q;\ x := 1\ \} \\
   & \qquad\qquad}{x := 0} }\ \}} \\
 \end{aligned} \label[prog]{prog: indicator version two state automaton}
\end{align}

To put this theory to work, we implemented an equivalence checker for CF-GKAT\@, with a front-end that can convert C code to CF-GKAT expressions by leveraging Clang's parser.
The resulting tool is able to automatically validate highly non-trivial program transformations like Erosa and Hendren's classic \command{goto}-elimination procedure~\cite{erosa-hendren-1994} and the control flow structuring of Ghidra.\footnote{\label{footnote:Ghidra}\url{https://ghidra-sre.org/}}

\paragraph{Outline}
The remainder of this article is organized as follows.
In \Cref{sec:overview}, we give an overview of CF-GKAT, including its syntax, continuation semantics, and trace semantics.
In \Cref{section:decision procedure}, we propose an automaton model for the continuation semantics, called CF-GKAT automata; we show how to translate a CF-GKAT expression to a CF-GKAT automaton, which in turn can be lowered to a GKAT automaton; ultimately, this gives rise to an algorithm for checking trace equivalence of CF-GKAT expressions.
In \Cref{sec:experiments}, we report on an implementation of our algorithm, along with several experiments.
We discuss related work in \Cref{sec:related-work}, and conclude in \Cref{sec:conclusion}.

Our proofs are formalized~\cite{artifact} in Coq~\cite{coq-web}; we provide proof sketches for the sake of intuition.

\section{Overview}%
\label{sec:overview}
%FIXME: CZ: multiple indicator variable?

In this section, we introduce the language of CF-GKAT, and gradually develop its semantics.
We begin by explaining the syntax of CF-GKAT;\@ after that, we delve into the semantics of its tests.
We then introduce the intermediate semantic model of \emph{(labeled and indexed families of) guarded languages with continuations}, which is flattened into to a model based on \emph{guarded languages}.
Having defined these tools, we then conclude by giving a semantics to CF-GKAT programs in this model.
Along the way, we will single out and explain some of the finer points using examples.

% We fix a single indicator variable $x$, as well as a finite set $I$ of possible indicator values.

\subsection{Syntax}

The syntax of CF-GKAT consists of two levels, similar to GKAT\@.
At the bottom level, there are \emph{tests}.
These are Boolean assertions that can occur as guards inside conditional statements, or within assertions that occur in the program text.
To model them, we fix a finite set of primitive tests $T$, which represent uninterpreted expressions that may or may not hold.
The full syntax is as follows.
\[
 \BExp  \ni  b, c ::=
 \false
 \mid \true
 \mid t  \in  T
 \mid {\color{blue}x = i\ (i  \in  I)}
 \mid b  \lor  c
 \mid b  \land  c
 \mid  \neg  b
\]
Compared to GKAT, tests in CF-GKAT include the \emph{indicator variable test} $x = i$ (highlighted in \textcolor{blue}{blue}) for each \emph{indicator value} $i$ drawn from a finite but fixed set of possible indicator values $I$.
As the notation suggests, this test holds when the indicator variable $x$ currently has the value $i$.
% Since the set \(I\) is finite, any complex predicate \(P\) on \(I\) can be encoded as a disjunction that enumerates all the values in \(P\)\ : \[ \bigvee  \{(x = i)  \mid  P(i)\}.\]

The top level syntax of GKAT is built using a finite set of uninterpreted commands \( \Sigma \) (the \emph{primitive actions}), as well as \emph{assertions} of the form $\comAssert{b}$, where $b \in \BExp$ is a test.
Expressions are composed using sequencing, \texttt{if} statements, and \texttt{while} loops.
CF-GKAT extends the base elements of the syntax with indicator variable assignments \(x := i\) (for each $i \in I$), which changes the value of the indicator variable \(x\) to \(i\).
In addition, it adds the non-local control flow commands $\comBrk$ and $\comRet$, as well as $\comGoto{\ell}$ and $\comLabel{\ell}$, where $\ell$ is taken from a fixed but finite set of labels $L$.
The full syntax is given below (additions relative to GKAT highlighted in \textcolor{blue}{blue} again).
\begin{align*}
 \Exp  \ni  e, f ::= {}&
 \comAssert{b}
  \mid  p  \in   \Sigma
  \mid  {\color{blue}x := i\ (i  \in  I)}
  \mid  e; f
  \mid  \comITE{b}{e}{f}  \mid  {} \\
 &
 \comWhile{b}{e}
  \mid  {\color{blue} \comBrk}
  \mid  {\color{blue} \comRet}
  \mid  {\color{blue} \comGoto{ \ell }\ ( \ell   \in  L)}
  \mid  {\color{blue} \comLabel{ \ell }\ ( \ell   \in  L)}
\end{align*}

% \begin{remark}
% We do not specify the sets $T$, $\Sigma$, $I$ and $L$ of (respectively) primitive tests, primitive actions, indicator values, and labels in more detail on purpose.
% When comparing CF-GKAT expressions, it is more practical to gather the symbols that appear in the corresponding positions, and construct the relevant set from them.
% For example, when we compare \Cref{prog: assignment inside branches} to \Cref{prog: assignment outside branches}, we infer that $T = \emptyset$, $\Sigma = \{ p, q \}$, $I = \{ 0, 42 \}$, and $L = \emptyset$; similarly, when comparing \Cref{prog: goto version two state automaton} to \Cref{prog: break version two state automaton}, we set $T = \{ t \}$, $\Sigma = \{ p, q \}$, $I = \emptyset$ and $L = \{ \ell_0, \ell_1 \}$.
% This is sound, because including more symbols in these sets does change whether or not the two programs are actually equivalent.
% \end{remark}

A \emph{valid program}, or \emph{program} for short, is an expression without (1)~duplicate labels, (2)~$\command{goto}$ commands with an undefined label, or (3)~$\comBrk$ statements that occur outside a loop.
For the sake of simplicity, we assume that the reader does not require a more formal definition of this notion.

\begin{example}
 Any GKAT expression is a valid program.
 Also, \cref{prog: break version two state automaton,%
  prog: goto version two state automaton,%
  prog: indicator version two state automaton}
 from the introduction are all valid CF-GKAT programs.
 The following expressions are \emph{not} valid programs:
 \begin{align*}
  \comLabel{ \ell }; (\comITE{t}{\comLabel{ \ell }; p}{q}) \tag{the label \( \ell \) is defined twice} \\
  (\comWhile{\true}{p}); \comGoto{ \ell } \tag{the label \( \ell \) is undefined} \\
  \comITE{t}{\comBrk}{p} \tag{$\comBrk$ appears outside a loop}
 \end{align*}
\end{example}

\subsection{Boolean Semantics}

To assign a semantics to CF-GKAT expressions, we first need to talk about the semantics of tests.
Intuitively, each test in a CF-GKAT expression symbolically denotes a set of execution contexts in which it is true.
But how do we model an execution context?
Because the primitive tests from $T$ are uninterpreted, we represent them by simply listing the ones that are true; the unlisted tests are then assumed to be false.
This is in line with how the semantics of (G)KAT handles tests~\cite{smolka_GuardedKleeneAlgebra_2020,kozen_KleeneAlgebraTests_1997}. % chktex 36
As for the indicator tests, we include the current value of the indicator variable $x$ in the execution context.
This means that each execution context is of the form $(i,  \alpha )$, for $i  \in  I$ and $ \alpha   \subseteq  T$.

Putting these ideas together, we can calculate the set of execution contexts that satisfy a given test \(b  \in  \BExp\) by induction.
This set will be regarded as the semantics of \(b\).

\begin{definition}
 Let $\At$ denote $2^T$, the set of \emph{atoms} of (the free Boolean algebra generated by) $T$.
 We define the \emph{Boolean semantics} function $ \lBrack - \rBrack : \BExp  \to  2^{I  \times  \At}$ inductively,
 as follows.
 \begin{align*}
   \lBrack  \false  \rBrack  &  \triangleq   \emptyset
    &  \lBrack  t  \rBrack &  \triangleq  \{ (i,  \alpha )  \mid  i  \in  I, t  \in   \alpha  \}
    &  \lBrack  b  \lor  c  \rBrack  &  \triangleq   \lBrack  b  \rBrack   \cup   \lBrack  c  \rBrack  \\
   \lBrack  \true  \rBrack   &  \triangleq  I  \times  \At
    &  \lBrack  x = i  \rBrack  &  \triangleq  \{ (i,  \alpha )  \mid   \alpha   \in  \At \}
    &  \lBrack  b  \land  c  \rBrack  &  \triangleq   \lBrack  b  \rBrack   \cap   \lBrack  c  \rBrack  \\
  & & & &  \lBrack   \neg  b  \rBrack  &  \triangleq  I  \times  \At \setminus  \lBrack  b  \rBrack
 \end{align*}
\end{definition}

\begin{example}
 Take $T = \{ t_1, t_2 \}$ and $I = \{ 1, 2, 3 \}$; then we can calculate that
 \[
   \lBrack  (t_1  \lor   \neg  t_2)  \land  (x = 2)  \rBrack  = \{
  (\{ t_1, t_2\}, 2),
  (\{ t_1 \}, 2),
  ( \emptyset , 2)
  \}
 \]
 In other words, the test above holds in execution contexts where the indicator has value $2$, and either $t_1$ and $t_2$ are both true (first element), both false (last element), or where $t_{1}$ is true but $t_{2}$ is false (middle element).
 In contrast, $ \lBrack  x = 1  \land  x = 3 \rBrack  =  \emptyset $, which is to say that this test does not hold in any execution context, because the indicator variable \(x\) cannot be both 1 and 3 at the same time.
\end{example}

\subsection{Guarded Language with Continuations}\label{sec:continuation-semantics}

We can now turn our attention to the semantics of CF-GKAT\@.
Like (G)KAT, the semantics of CF-GKAT is given in terms of \emph{guarded languages}~\cite{smolka_GuardedKleeneAlgebra_2020,kozen_KleeneAlgebraTests_1997}, which are best thought of as sets of symbolic traces of the program. % chktex 36
These traces record the initial, intermediate and final machine states observed during execution, as well as the actions that occurred between those states. % chktex 36
Because the value of the indicator variable matters only for control flow, we do not consider indicators to be part of the machine state; hence, machine states in a guarded word are drawn from $\At$.
\begin{definition}
 A \emph{guarded word} is a sequence of the form $ \alpha _{1} p_{1}  \alpha _{2} p_{2}  \cdots   \alpha _{n-1} p_{n}  \alpha _{n}$, where \( \alpha _{i}  \in  \At\) and \(p_{i}  \in   \Sigma \); that is to say, guarded words are elements of the regular language $\At  \cdot  ( \Sigma   \cdot  \At)^*$.
 We refer to sets of guarded words as \emph{guarded languages}; the set of guarded languages is denoted $\mathcal{G}$.
\end{definition}
\begin{example}
 Let $T = \{ t_{1}, t_{2} \}$ and $\Sigma = \{ p_{1}, p_{2} \}$.
 Now the guarded word $\{ t_{1} \}p_{1}\{ t_2 \}p_{2}\emptyset$ represents a program trace that starts out in a machine state where $t_1$ is true (but $t_{2}$ is not).
 The program then executes the action $p_{1}$, after which $t_2$ is true (but $t_{1}$ is not).
 Finally, the program goes on to execute the action $p_{2}$, and halts in a state where neither $t_{1}$ nor $t_{2}$ is true.
\end{example}

Our semantics of a CF-GKAT expression will ultimately be a guarded language.
However, due to the non-local nature of our language, we cannot compute this semantics inductively, as the label in \command{goto} may occur in a different part of the program whose traces have not yet been computed.
To address this challenge, we introduce an intermediate \emph{continuation semantics}, which accounts for both the indicator variables and the non-local control flow.
The remainder of this subsection explains the domain of continuation semantics, based on \emph{guarded words with continuations}.
Intuitively, these are guarded words equipped with a piece of information called a \emph{continuation}, which indicates how control flow continues after the program ends.
The inclusion continuation information at the end of a trace allows us to define a semantics of CF-GKAT expressions inductively.
Once the continuation semantics of a CF-GKAT program is known, we can flatten it into a guarded language.

\begin{definition}%
\label{def:guarded-word-with-continuation}
 A \emph{guarded word with continuation} is a pair $w  \cdot  c$,
 where $w$ is a guarded word and $c$ is a \emph{continuation},
 which can take on one of the following forms for $i  \in  I$ and $ \ell   \in  L$:
 \begin{mathpar}
  \acc{i} \and
  \brk{i} \and
  \ret \and
  \jmp{( \ell , i)}
 \end{mathpar}
 We write $C$ for the set of all continuations.
 A set of guarded words with continuations is a \emph{guarded language with continuations}; the set of guarded languages with continuations is written $\mathcal{C}$.
\end{definition}
Intuitively, the different types of continuation may be interpreted as follows:
\begin{itemize}
    \item
    The continuation $\acc{i}$ represents that the trace has successfully reached the end of this part of the program, with indicator value \(i\).
    Execution can be picked up if the program is put in a larger context --- e.g., if $w \cdot \acc{i}$ is a trace of $e$, then it may be combined with a trace found when $f$ is executed with indicator value $i$ to compute the semantics of $e; f$.
    \item
    A continuation of the form $\brk{i}$ signals that the trace ends by halting the loop in which it occurs.
    Execution can resume only after this loop (with indicator value $i$).
    This kind of trace cannot be composed on the right, as is done for traces with accepting continuations, because we first need to enclose it in a loop to halt; it will then be converted into $\acc{i}$.
    \item
    The continuation $\ret$ represents a trace that ends in the program halting completely.
    Traces of this kind will percolate upwards in the semantics, without changing their continuation.
    These are intended to model the $\comRet$ statement, which halts the program no matter how deeply it is nested.
    In this case, the indicator value does not matter any more.
    \item
    Finally, the continuation $\jmp{(\ell, i)}$ is put on traces that will continue executing from label $\ell$, with indicator value $i$.
    Like $\brk{i}$ and $\ret$, these traces do not compose on the right, but unlike $\brk{i}$ this continuation does not change, as jump resolution happens only at the end, when the continuation semantics is known for the entire program.
\end{itemize}

\begin{example}
 Let $w$ be the guarded word from the previous example;
 the guarded word with continuation $w  \cdot  \jmp{( \ell _{1}, 2)}$
 represents a partial program trace that takes the steps in $w$,
 continues executing at the label $ \ell _{1}$ with indicator value $2$.
 More examples appear in the next section.
\end{example}

\subsection{Indexed Families and Sequencing}
The continuation semantics of a CF-GKAT expression takes a starting indicator value, and produces a guarded language with continuations representing the traces of that program when started with the given indicator value.
To properly encode this, we need the following notion.
\begin{definition}
An \emph{indexed family} of guarded languages (with continuations), or ``indexed family'' for the sake of brevity, is function from $I$ to guarded languages (with continuations).
Typically, we use \(G\) and \(H\) to denote an indexed family.
To lighten notation, we write \(G_{i}\) to denote \(G(i)\).
\end{definition}

Similar to guarded languages, indexed families can be composed in several ways.
In particular, we are interested in the sequencing operation and the Kleene star operation of indexed families, because these will turn out to be useful when defining the continuation semantics of CF-GKAT\@.

When sequencing two families \(G\) and \(H\), the traces in \(G_i\) with a continuation of the form $\acc{j}$ will be composed with traces in \(H_j\); traces with different continuations are copied over in full, because they do not compose on the right.
Formally, this operation is defined as follows.

\begin{definition}%
\label{def:sequencing}
 Let $G, H: I \to \mathcal{C}$.
 We write $G  \diamond  H$ for the \emph{sequencing} (or \emph{concatenation}) operation of \(G\) and \(H\), which is defined as the smallest family of guarded languages with continuations (in the pointwise order) satisfying the following rules for all $i,j  \in  I$ as well as all $ \ell   \in  L$:
 \begin{mathpar}
  \inferrule{%
   w \alpha   \cdot  \acc{j}  \in  G_{i} \\
    \alpha x  \cdot  c  \in  H_{j}
  }{%
   w \alpha x  \cdot  c  \in  (G  \diamond  H)_{i}
  }
  \and
  \inferrule{%
   w  \cdot  c  \in  G_{j} \\
   c \in \{ \brk{i}, \acc{i}, \jmp{(\ell, i)} \}
  }{%
   w  \cdot  c  \in  (G  \diamond  H)_{j}
  }
 \end{mathpar}
\end{definition}
The first rule composes accepting traces in $G$ with traces in $H$, picking up with the indicator value where the first trace left off.
Note also that this rule requires the last atom in the trace on the left to match the first atom in the trace on the right, because we want the second trace to start from the machine state computed in the first trace.
This mirrors the \emph{coalesced product} used to define the sequential composition of guarded languages (without continuations)~\cite{kozen_KleeneAlgebraTests_1997}.
The last rule ensures that traces that encountered non-local control flow within $G$ are preserved in $G  \diamond  H$.

\begin{example}%
\label{example:sequencing}
 Let $I = \{1,2\}$, and let $G$ and $H$ be indexed families given by:
 \begin{align*}
  G_{1} & = \{  \alpha p \beta   \cdot  \brk{1},\;  \beta p \alpha   \cdot  \acc{2} \}
    & G_{2}& = \{  \alpha q \beta  \cdot \acc{1} \} \\
  H_{1} & = \{  \gamma q \beta   \cdot  \ret \}
    & H_{2} & = \{  \alpha r \beta   \cdot  \jmp{( \ell _{1}, 1)} \}
  \intertext{
   Then we can compute that the sequencing $G  \diamond  H$ is the following indexed family:
  }
  (G  \diamond  H)_{1} & = \{  \alpha p \beta   \cdot  \brk{1},\;  \beta p \alpha r \beta   \cdot  \jmp{( \ell _{1}, 1)} \}
           & (G  \diamond  H)_{2} & =  \emptyset
 \end{align*}
 Here, $(G \diamond H)_1$ contains $\alpha p \beta \cdot \brk{1}$ by the second rule, because $G_1$ does.
 Furthermore, the trace $\beta p \alpha \cdot \acc{2}$ in $G_1$ is composed with $\alpha r \beta \cdot \jmp{(\ell_1, 1)}$ from $H_2$ to form $\beta p \alpha r \beta \cdot \jmp {(\ell_1, 1)}$ in $(G \diamond H)_1$, by the first rule.
 The set $(G \diamond H)_2$ is empty, because despite the fact that $ \alpha q \beta  \cdot \acc{1} \in G_2$, there is no trace in $H_1$ that starts with $\beta$, and so neither rule can apply.
\end{example}

\subsection{Continuation Semantics}
With the theory of indexed families in place, we can now define the continuation semantics $ \lBrack e \rBrack ^ \sharp $ of a CF-GKAT program $e$ in terms of an indexed family.
We start with the base cases.

\begin{definition}[Continuation semantics, base]
 For all $i, j  \in  I$, we define the following sets:
 \begin{align*}
   \lBrack \comAssert{b} \rBrack _{i}^ \sharp  &  \triangleq  \{  \alpha   \cdot  \acc{i}  \mid  (i,  \alpha )  \in   \lBrack  b  \rBrack  \}
    &  \lBrack  \comGoto{ \ell }  \rBrack _{i}^ \sharp  &  \triangleq  \{  \alpha   \cdot  \jmp{( \ell , i)}  \mid   \alpha   \in  \At \} \\
   \lBrack p \rBrack _{i}^ \sharp              &  \triangleq  \{ \alpha  p  \beta   \cdot  \mathbf{acc}\ i  \mid   \alpha ,  \beta   \in  \At\}
    &  \lBrack  \comLabel{ \ell }  \rBrack _{i}^ \sharp  &  \triangleq  \{  \alpha   \cdot  \acc{i}  \mid   \alpha   \in  \At \} \\
   \lBrack  x := j  \rBrack _{i}^ \sharp  &  \triangleq  \{  \alpha   \cdot  \acc{j}  \mid   \alpha   \in  \At \}
    &  \lBrack  \comBrk  \rBrack _{i}^ \sharp  &  \triangleq  \{  \alpha   \cdot  \brk{i}  \mid   \alpha   \in  \At \} \\
   \lBrack  \comRet  \rBrack _{i}^ \sharp  &  \triangleq  \{  \alpha   \cdot  \ret  \mid   \alpha   \in  \At \}
 \end{align*}
 \end{definition}

Each of these base syntax elements yields a finite indexed family.
For the constructs $\command{return}$, $\command{goto}$, and $\command{break}$, all traces terminate immediately in the corresponding continuation.

We inherit the semantics of assertions and primitive actions from (G)KAT~\cite{kozen_KleeneAlgebraTests_1997,smolka_GuardedKleeneAlgebra_2020}. % chktex 36
Assertions have traces that accept when their only atom satisfies the test.
A primitive action $p$ has traces of the form $\alpha p \beta \cdot \acc{i}$ for all $\alpha, \beta \in \At$: because $p$ is uninterpreted, we can reach any machine state by running $p$.
Only the indicator value is retained, because primitive actions cannot interact with indicators.
Assignments like $x := j$ accept immediately, without changing the machine state; however, each trace ends with indicator value $j$ --- regardless of the initial indicator value $i$.

Finally, labels are encoded as no-operations, which makes them neutral for sequencing operator, i.e., we have $ \lBrack  \comLabel{ \ell }  \rBrack ^ \sharp   \diamond  G = G =  \lBrack  \comLabel{ \ell }  \rBrack ^ \sharp   \diamond  G$ for all indexed families $G$.
This is because labels serve only as potential starting points of execution; we will leverage them in the next subsection.

\begin{remark}
  For soundness, it is important that the indicator variable $x$ does not occur in any primitive test $t  \in  T$ or action $p  \in   \Sigma $.
  Concretely, primitive tests do not restrict the current value indicator variable, and primitive actions preserve the indicator value after their executions.
\end{remark}

We now turn our attention to the program composition operators.
These are generalizations of the guarded language semantics of GKAT~\cite{smolka_GuardedKleeneAlgebra_2020}.
First of all, the $\comITE{b}{e}{f}$ filters out traces in the semantics of the $e$ that satisfy the guard $b$, as well as the traces in $f$ that invalidate it.

 \begin{definition}[Continuation semantics, branching]
 Let $e, f \in \Exp$.
 We set $ \lBrack  \comITE{b}{e}{f}  \rBrack ^ \sharp $ to be the least indexed family that satisfies the following rules for all $i \in I$:
 \begin{mathpar}
    \inferrule{
        (i, \alpha) \in  \lBrack  b  \rBrack  \\
        \alpha{}w \cdot c \in  \lBrack  e  \rBrack ^ \sharp _{i}
    }{%
        \alpha{}w \cdot c \in  \lBrack  \comITE{b}{e}{f}  \rBrack ^ \sharp _{i}
    }
    \and
    \inferrule{
        (i, \alpha) \not\in  \lBrack  b  \rBrack  \\
        \alpha{}w \cdot c \in  \lBrack  f  \rBrack ^ \sharp _{i}
    }{%
        \alpha{}w \cdot c \in  \lBrack  \comITE{b}{e}{f}  \rBrack ^ \sharp _{i}
    }
 \end{mathpar}
 \end{definition}

 The semantics of the sequencing operator is easy: it just composes the semantics of the operands with the sequencing operator we have for indexed families.
 For loops, some more care is needed because traces can be iterated, and we need to account for early termination.

 \begin{definition}[Continuation semantics, sequencing and loops]%
 \label{def:intermediate-sequencing-loops}
 Let $e, f \in \Exp$ and $b \in \BExp$.
 We define
 \(
     \lBrack e; f \rBrack ^ \sharp   \triangleq   \lBrack  e  \rBrack ^ \sharp   \diamond   \lBrack  f  \rBrack ^ \sharp
 \).
 Also, $ \lBrack \comWhile{b}{e} \rBrack ^ \sharp $ is the least indexed family satisfying for $i \in I$:
 \begin{mathpar}
    \inferrule{%
        (i, \alpha) \not\in  \lBrack  b  \rBrack
    }{%
        \alpha \cdot \acc{i} \in  \lBrack \comWhile{b}{e} \rBrack ^ \sharp _{i}
    }
    \and
    \inferrule{%
        (i, \alpha) \in  \lBrack  b  \rBrack  \\
        \alpha w \cdot c \in ( \lBrack  e  \rBrack ^ \sharp   \diamond   \lBrack \comWhile{b}{e} \rBrack ^ \sharp )_{i}
    }{%
        \alpha w \cdot \lfloor c \rfloor \in  \lBrack \comWhile{b}{e} \rBrack ^ \sharp _i
    }
 \end{mathpar}
 The operation $\lfloor - \rfloor$ in the last rule is defined by $\lfloor c \rfloor = \acc{i}$ when $c = \brk{i}$, and $\lfloor c \rfloor = c$ otherwise.
\end{definition}

%FIXME: We need some kind of example around here, preferably based on a program from the introduction.

The first rule accounts for traces that halt immediately because the loop guard is false.
The second rule allows prepending traces from the loop body that satisfy the guard.
Because of the way sequencing works, body traces that end in $\brk{i}$ may occur; the second rule converts their continuations to $\acc{i}$, signaling that the loop has been exited and normal control flow can resume.

\smallskip
The semantics we have so far defines the traces of a program starting from the beginning.
However, a CF-GKAT program can also resume from a label when a \command{goto} command is encountered.
To obtain the traces for a given label $\ell$, we must descend into the program until we encounter \(\comLabel{ \ell }\).
For the base cases, this is relatively simple to accomplish: just check if we start at the label.

\begin{definition}[Continuation semantics starting from a label, base]
 Let $e  \in  \Exp$.
 For each $ \ell   \in  L$, we define the following guarded languages with continuations:
 \begin{mathpar}
   \lBrack \comAssert{b} \rBrack _{i}^ \ell    \triangleq   \emptyset  \and
   \lBrack  \comGoto{ \ell '}  \rBrack _{i}^ \ell    \triangleq   \emptyset  \and
   \lBrack p \rBrack _{i}^ \ell   \triangleq   \emptyset  \and
   \lBrack x := j \rBrack _{i}^ \ell   \triangleq   \emptyset  \and
   \lBrack  \comBrk  \rBrack _{i}^ \ell   \triangleq   \emptyset  \and
   \lBrack  \comRet  \rBrack _{i}^ \ell   \triangleq   \emptyset  \and
   \lBrack  \comLabel{ \ell '}  \rBrack _{i}^ \ell   \triangleq  \{  \alpha   \cdot  \mathbf{acc}\ i  \mid   \alpha   \in  \At,\  \ell  =  \ell ' \}
 \end{mathpar}
\end{definition}
Note how none of these cases has a trace, except the one for $ \lBrack  \comLabel{ \ell '}  \rBrack _{i}^ \ell $ when $\ell' = \ell$, which accepts immediately.
With these cases covered, we can then treat the inductive step.

\begin{definition}[Continuation semantics starting from a label, sequencing and branching]
Let $e, f \in \Exp$, $b \in \BExp$ and $\ell \in L$.
We define the following indexed families to cover the traces of CF-GKAT programs starting from the label $\ell$ when composed using branching or sequencing:
\begin{mathpar}
     \lBrack \comITE{b}{e}{f} \rBrack ^ \ell _{i}  \triangleq   \lBrack  e  \rBrack ^ \ell _{i} \cup  \lBrack  f  \rBrack ^ \ell _{i}
    \and
     \lBrack e; f \rBrack ^ \ell _{i}  \triangleq  ( \lBrack  e  \rBrack ^ \ell   \diamond   \lBrack  f  \rBrack ^ \sharp )_{i} \cup  \lBrack  f  \rBrack ^ \ell _{i}
\end{mathpar}
\end{definition}
For branching, the semantics starting from $\ell$ disregards the guard and descends into the operands.
The sequencing case is more interesting: here, we still need to account for the traces that start from the beginning of $f$ after executing a trace in $e$ starting from the label $ \ell $.

The only remaining case to cover is the loop.
If we start execution from a label somewhere in the loop body, we may need to start the loop again after completing the loop body.
On the other hand, early termination in the loop body still needs to be turned into an accepting trace.

\begin{definition}[Continuation semantics starting from a label, loops]
Let $e \in \Exp$ and $b \in \BExp$.
We define the indexed family $ \lBrack \comWhile{b}{e} \rBrack ^ \ell $ below, where $\lfloor - \rfloor$ is as in \Cref{def:intermediate-sequencing-loops}:
\[
     \lBrack \comWhile{b}{e} \rBrack ^ \ell _i = \{ w \cdot \lfloor c \rfloor \mid w \cdot c \in ( \lBrack  e  \rBrack ^ \ell   \diamond   \lBrack \comWhile{b}{e} \rBrack ^ \sharp )_{i} \}
\]
\end{definition}

\subsection{Guarded Language Semantics}

The continuation semantics $ \lBrack  e  \rBrack ^ \sharp $ of a CF-GKAT program uses continuations to record how a trace ends.
In particular, some traces may end in a continuation of the form $\jmp{(\ell, i)}$, signaling that computation needs to continue from the label $\ell$.
The semantics \( \lBrack  e  \rBrack ^ \ell \) provides the necessary information to resolve such jump continuations: we can resume \(\jmp{( \ell , i)}\) by concatenating with the traces in \( \lBrack e \rBrack ^ \ell _{i}\).
We will end this section by doing just that.

To formalize our approach, we need a way to refer to the continuation CF-GKAT semantics of a program as a whole, i.e., for all indicator values, starting from either the beginning or some label.

\begin{definition}
 A \emph{labeled family of guarded languages (with continuations)}, or \emph{labeled family} for short, is a function $G$ from $L +  \sharp $ to indexed families of guarded languages (with continuations), e.g., $G: L +  \sharp   \to  I  \to  \mathcal{C}$.
 We often use superscripts to denote the value at a given label $ \sharp $, writing $G^ \ell $ for $G( \ell )$.
 Note that under this convention, $G^ \ell $ is an indexed family, which means that we may further unravel by writing $G^ \ell _i$ to obtain the guarded language with continuations $G( \ell , i)$.
\end{definition}

Crucially, we can retrofit the continuation semantics $ \lBrack  e  \rBrack $ as a labeled family; after all, $ \lBrack  e  \rBrack ^\sharp$ is an indexed family, and so is $ \lBrack  e  \rBrack ^\ell$ for each $\ell \in L$.
We will treat $ \lBrack  e  \rBrack $ as such from this point on.

\smallskip
To resolve the jumps in a labeled family of guarded languages with continuations, we resolve the traces ending in $\jmp{( \ell , i)}$ by looking up the traces that originate from label $ \ell $ with indicator value $i$.
%FIXME: CZ: I think this paragraph starts with "resolve jumps", but the rest is describing the lowering operation. Confused me for a bit. Maybe we can mention lowering either at the beginning of the paragraph or here.
% For example, I was not sure what "remove" continuation acc and ret means in the context of jump resolution.
We also remove the continuations $\acc{i}$ and $\ret$, because those come with traces that either reached the end of the program, or encountered a $\comRet$ statement respectively.
Continuations of the form $\brk{i}$ should not occur at the top level when computing the semantics of a program, so we can ignore them.
The result is a labeled family of guarded languages (without continuations).

\begin{definition}
 Let $G: L +  \sharp   \to  I  \to  \mathcal{C}$ be a labeled family of guarded languages with continuations.
 We write $G\! \downarrow $ for the (point-wise) least labeled family of guarded languages $G\!\downarrow$ such that the following rules are satisfied for all $k  \in  L +  \sharp $, $ \ell   \in  L$, and $i, j  \in  I$:
 \begin{mathpar}
  \inferrule{%
   w  \cdot  \acc{i}  \in  G_{i}^{k}
  }{%
   w  \in  G\! \downarrow _{i}^{k}
  }
  \and
  \inferrule{%
   w  \cdot  \ret  \in  G_{i}^{k}
  }{%
   w  \in  G\! \downarrow _{i}^{k}
  }
  \and
  \inferrule{%
   w \alpha   \cdot  \jmp{( \ell , j)}  \in  G_{i}^{k} \\
    \alpha x  \in  G\! \downarrow _{j}^ \ell
  }{%
   w \alpha x  \in  G\! \downarrow _{i}^{k}
  }
 \end{mathpar}
\end{definition}
The first two rules take care of flattening guarded words with continuations that end in acceptance.
The third rule strings together guarded words continuations that jump to a different label.

\begin{example}
 Let $G$ be the labeled family of guarded languages with continuations defined by
 \begin{align*}
  G_{1}^ \sharp  & = \{  \alpha   \cdot  \jmp{(\ell, 1)} \}
    & G_{2}^ \sharp  & =  \emptyset  \\
  G_{1}^{\ell} & = \{  \alpha  p  \alpha   \cdot  \jmp{(\ell', 1)},  \beta   \cdot  \acc{1} \}
    & G_{2}^{\ell} & = \{  \alpha   \cdot  \jmp{(\ell', 2)} \} \\
  G_{1}^{\ell'} & = \{  \alpha  q  \alpha   \cdot  \jmp{(\ell, 1)},  \alpha  r  \beta   \cdot  \jmp{(\ell, 1)} \}
    & G_{2}^{\ell'} & = \{  \alpha   \cdot  \jmp{(\ell, 1)}\}
 \end{align*}
 Now $G\! \downarrow _{1}^ \sharp $ contains, among other things,
 the guarded word $ \alpha  p  \alpha  q  \alpha  p  \alpha  r  \beta $.
 Furthermore, $G\! \downarrow _{2}^{\ell}$ is empty, despite $G_{2}^{\ell}$ and $G_{2}^{\ell'}$ containing guarded words with mutually jumping continuations, as these can never be concatenated into one guarded word with continuation of the form $\acc{i}$ or $\ret$.
\end{example}

In total, we can then obtain the semantics of a CF-GKAT term as $ \lBrack  e  \rBrack \!\downarrow$, in the form of a labeled family of guarded languages.
This concludes our discussion of the semantics of CF-GKAT\@.

% \begin{lemma}\label{the: label missing causes empty semantics}
%     If \(\comLabel{ \ell }\) does not appear in expression \(e\), then \( \forall  i  \in  I,  \lBrack e \rBrack _{i}^ \ell  =  \emptyset \).
% \end{lemma}

% !TEX root=../paper.tex

\section{Decision Procedure}%
\label{section:decision procedure}

To decide whether two CF-GKAT expressions denote the same guarded language, we propose CF-GKAT automata.
We will show that every CF-GKAT expression can be converted to a CF-GKAT automaton with the same continuation semantics, by induction on the expression \emph{à la} Thompson~\cite{thompson_ProgrammingTechniquesRegular_1968}.
However, because CF-GKAT automata implement the continuation semantics, directly comparing them for language equivalence will not be sufficient.
For instance, the following programs exhibit identical trace semantics, but differ in their continuation semantics:
\begin{mathpar}
  x := 1 \and \comAssert{\true}.
\end{mathpar}
The first program sets the indicator variable to 1, and the second program simply does nothing.
These programs have the same trace semantics, because the trace semantics does not care about the final value of the indicator variable.
Yet, their continuation semantics are different: guarded words in \( \lBrack x := 1 \rBrack _{i}^ \sharp \) are always paired with the continuation \(\acc{1}\) regardless of \(i\), but \( \lBrack \comAssert{\true} \rBrack _{i}^ \sharp \) preserves the starting indicator by outputting the continuation \(\acc{i}\).
Furthermore, equivalence checking needs to account for jump resolution, by looking up the sequel of a trace ending in a continuation of the form $\jmp (i, \ell)$ in the traces starting from label $\ell$ with indicator value $i$.

Thus, our decision procedure does not check equivalence of CF-GKAT automata directly; instead, we \emph{lower} CF-GKAT automata into GKAT automata in a way that mirrors the move from the continuation semantics to the guarded language semantics.
This lets us reuse the decision procedure for GKAT automata equivalence~\cite{smolka_GuardedKleeneAlgebra_2020}.
The soundness and completeness of our decision procedure follows from the correctness of Thompson's construction (\Cref{the:thompson-correctness}), the correctness of the lowering (\Cref{the:cf-gkat-automaton-lowering-correctness}), and finally the correctness of GKAT automata equivalence checking~\cite{smolka_GuardedKleeneAlgebra_2020}.

\subsection{GKAT Automata}

To establish our decision procedure, we will first recap GKAT automata and their trace semantics.
\begin{definition}[GKAT automata~\cite{smolka_GuardedKleeneAlgebra_2020, kozen_BohmJacopiniTheorem_2008a}]
 A \emph{GKAT automaton} $A  \triangleq   \langle  S,  \delta , \hat{s}  \rangle $ consists of a set of \emph{states} $S$, a \emph{transition function} $\delta: S  \to  \At  \to   \bot  +  \top  +  \Sigma   \times  S$, and a \emph{start state} $\hat{s}  \in  S$.
\end{definition}
Intuitively, given a state $s$ and an atom $ \alpha $ accounting for the truth value of each primitive test, a GKAT automaton will either \emph{reject} the input, represented by $ \delta (s,  \alpha ) =  \bot $; \emph{accept} the input, represented by $ \delta (s,  \alpha ) =  \top $; or to \emph{transition} to a new state in $S$ after executing an action from $ \Sigma $, represented by $ \delta (s,  \alpha )  \in   \Sigma   \times  S$.
A GKAT automaton induces a guarded language, by tracing all the possible execution paths reaching an accepting transition.

\begin{definition}
 Given a GKAT automaton $A  \triangleq   \langle  S,  \delta , \hat{s}  \rangle $, we define $ \lBrack  -  \rBrack _A: S  \to  \mathcal{G}$ as the (pointwise) smallest function satisfying the following rules for all $s  \in  S$ and $ \alpha   \in  \At$:
 \begin{mathpar}
  \inferrule{%
    \delta (s,  \alpha ) =  \top
  }{%
    \alpha   \in   \lBrack  s  \rBrack _A
  }
  \and
  \inferrule{%
    \delta (s,  \alpha ) = (p, s') \\
   w  \in   \lBrack  s'  \rBrack _A
  }{%
    \alpha pw  \in   \lBrack  s  \rBrack _A
  }
 \end{mathpar}
 Finally, we define the guarded language semantics of $A$ by setting $ \lBrack  A  \rBrack  =  \lBrack  \hat{s}  \rBrack _A$.
\end{definition}

The trace equivalence of finite GKAT automata is decidable, which we record as follows.
\begin{theorem}[Decidability for GKAT~\cite{smolka_GuardedKleeneAlgebra_2020}]
 Given two finite GKAT automata $A_{0}$ and $A_{1}$, it is decidable whether they represent the same guarded language, i.e., whether $ \lBrack  A_{0}  \rBrack  =  \lBrack  A_{1}  \rBrack $.
 The algorithm to do this has a complexity that is nearly-linear\footnote{$\mathcal{O}(n \cdot \hat{ \alpha }(n))$, where $\hat{ \alpha }$ is the inverse Ackermann function; c.f.~\cite{Tarjan75}.} in the total number of states.
\end{theorem}

\subsection{CF-GKAT Automata}

To leverage the efficient decision algorithm for GKAT automata, we will need to convert each CF-GKAT expression $e$ and a starting indicator value \(i\) into a GKAT automaton that recognizes the guarded language $ \lBrack  e  \rBrack \! \downarrow _{i}^\sharp$.
As discussed, this process is separated into two steps, where we use CF-GKAT automata as an intermediate between CF-GKAT expressions and GKAT automata.
Like GKAT automata, CF-GKAT automata need a transition structure.
Having a separate type for this transition structure will turn out to be useful, so we isolate it as follows.
\begin{definition}[CF-GKAT dynamics]
 Given a set $S$, we write $G(S)$ for the set of possible \emph{CF-GKAT dynamics} on $S$, which is given by the following function type, where $C$ is as in \Cref{def:guarded-word-with-continuation}:
 \[G(S)  \triangleq  I  \times  \At  \to   \bot  + C +  \Sigma   \times  X  \times  I.\]
 %FIXME: I think we need to fix the code to follow this type
\end{definition}

The elements of \(G(S)\) represent transitions exiting a single (initial) state or label in a CF-GKAT automaton over a state set $S$.
Given the current indicator value \(i  \in  I\) and an atom $ \alpha $ accounting for the truth value of each primitive test, a dynamics $ \rho   \in  G(S)$ may either: \emph{reject} the input ($ \rho (i,  \alpha ) =  \bot $); offer a \emph{continuation} ($ \rho (i,  \alpha )  \in  \mathcal{C}$); or an action, next state, and indicator value ($ \rho (i,  \alpha )  \in   \Sigma   \times  X  \times  I$).

Then the definition of CF-GKAT automaton is similar to GKAT automaton, except we will need a function \( \lambda : L  \to  G(S)\) where \( \lambda ( \ell )\) provides a dynamics representing the ``entry point'' of label \( \ell \).
\begin{definition}
 A \emph{CF-GKAT automaton} \(A  \triangleq   \langle S,  \delta , \hat{s},  \lambda  \rangle \) consists of a set of \emph{states} \(S\), a \emph{transition function} \( \delta : S  \to  G(S)\),
 a \emph{start state} \(\hat{s}  \in  S\), and a \emph{jump map} \( \lambda : L  \to  G(S)\).
\end{definition}

The transition map \( \delta \) assigns every state a dynamics, while the jump map $ \lambda $ assigns a dynamics to each label, indicating how to resume computation when a jump continuation is reached.

\begin{example}[A simple CF-GKAT automaton]%
\label{example:cf-gkat-automaton}
  Let $I = \{1,2\}$ and consider the following program:
  \[
    \comITE{t \wedge (x = 1)}{\{x := 2; \comLabel{ \ell };\; p\}}{\{x := 1;\; \comGoto{ \ell }\}}.
  \]
  We can construct the following automaton with the same continuation semantics from state $\hat{s}$:
  \[
    \begin{tikzpicture}
        \node (shat) {$\hat{s}$};
        \node[right=2cm of shat] (s) {$s$};
        \draw ($(shat) + (0mm, 7mm)$) edge[-latex] (shat);
        \draw (shat) edge[-latex] node[above] {\footnotesize $1,\{t\} \mid p, 2$} (s);
        \node[left=10mm of shat] (shatacc) {\footnotesize $\jmp (\ell, 1)$};
        \draw (shat) edge[double,double distance=2pt,-implies] node[above,align=center] {%
            \footnotesize $1, \emptyset$ \\[-1mm]
            \footnotesize $2, -$
        } (shatacc);
        \node[right=10mm of s] (sacc) {\footnotesize $\acc{i}$};
        \draw (s) edge[double,double distance=2pt,-implies] node[above] {\footnotesize $i, -$} (sacc);
    \end{tikzpicture}
  \]
  Here, \(\hat{s} \xrightarrow{1, \{t\} \mid p, 2} s\) means that \( \delta (\hat{s}, 1, \{t\}) = (p, s, 2)\) and \(s \xRightarrow{i, -} \acc{i}\) means that \( \delta (s, i, \alpha) = \acc{i}\) for $\alpha \in \At$ and $i \in I$.
  The behavior of this automaton, when starting from \(\hat{s}\), is as follows.
  \begin{itemize}
    \item
    If the input atom is \(\{t\}\) --- that is, the test $t$ is true --- and the indicator is $1$, we transition to the state \(s\), while setting the indicator to $2$ and executing \(p\) (and skipping over label $\ell$).
    Upon reaching $s$, the automaton accepts unconditionally, preserving the indicator.
    \item
    Otherwise, if the input atom is $\emptyset$ --- that is, the test $t$ is false --- or the indicator value is anything other than $1$, the automaton will simply offer the continuation \(\jmp{( \ell ,1)}\), thereby telling execution to continue from label $\ell$ with indicator value $1$.
  \end{itemize}
  As we can see, the behavior of \(\hat{s}\) indeed matches the behavior of the program when executing from the start.
  On the other hand, the entry dynamics for \( \ell \) can be defined as follows:
  \[ \forall  i  \in  I,  \alpha   \in  \At,  \lambda ( \ell , i,  \alpha )  \triangleq  (p, s, i).\]
  To put the above definition into words: when jumping to the label \( \ell \), we will reach the state \(s\) while executing \(p\).
  Thus, the behavior of \( \lambda ( \ell )\) matches the behavior of the program when starting from \( \ell \).
\end{example}

\begin{remark}
We could have opted to instrument CF-GKAT automata with a \emph{state} for each label $\ell$, rather than a dynamics.
This would simplify their definition, at the cost of complicating the Thompson construction (\Cref{sec:thompson-construction}).
An additional advantage of our approach is that we avoid creating unreachable states when lowering CF-GKAT automata to GKAT automata (\Cref{sec:lowering-cf-gkat-automata-to-gkat}).
\end{remark}

To formalize the intuition of ``behaviors'', as seen in the previous example, we will assign a continuation semantics to each CF-GKAT automaton.
It is convenient to first assign a semantics to each dynamics \( \rho   \in  G(S)\) in a CF-GKAT automaton \(A  \triangleq   \langle S,  \delta , \hat{s},  \lambda  \rangle \), instead of every state.

\begin{definition}[continuation semantics]
 Given an automaton \(A  \triangleq   \langle S,  \delta , \hat{s},  \lambda  \rangle \),
 the continuation semantics of each dynamics \( \rho   \in  G(S)\) is an indexed family \( \lBrack  \rho  \rBrack _A: I  \to  \mathcal{C}\), defined as the (point-wise) smallest set satisfying the following rules for $i, j  \in  I$, $ \alpha   \in  \At$ and $c \in C$:
 \begin{mathpar}
  \inferrule{%
    \rho (i,  \alpha ) = \acc{j}
  }{%
    \alpha   \cdot  \acc{j}  \in  ( \lBrack  \rho  \rBrack _A)_{i}
  }
  \and
  \inferrule{%
    \rho (i,  \alpha ) = \brk{j}
  }{%
    \alpha   \cdot  \brk{j}  \in  ( \lBrack  \rho  \rBrack _A)_{i}
  }
  \and
  \inferrule{%
    \rho (i,  \alpha ) = \ret
  }{%
    \alpha   \cdot  \ret  \in  ( \lBrack  \rho  \rBrack _A)_{i}
  }
  \\
  \inferrule{%
    \rho (i,  \alpha ) = \jmp{( \ell , j)}
  }{%
    \alpha   \cdot  \jmp{( \ell , j)}  \in  ( \lBrack  \rho  \rBrack _A)_{i}
  }
  \and
  \inferrule{%
    \rho (i,  \alpha ) = (p, s, j) \\
   w \cdot c  \in  ( \lBrack  \delta (s) \rBrack _A)_{j}
  }{%
    \alpha pw \cdot c  \in  ( \lBrack  \rho  \rBrack _A)_{i}
  }
 \end{mathpar}
 Similar to the continuation semantics of expressions, the continuation semantics of automata are also labeled families of guarded languages with continuations: the semantics from the start \( \lBrack A \rBrack ^ \sharp \) is defined by the dynamics of the start state, and that of a label \( \ell   \in  L\) is defined by the jump map:
 \begin{mathpar}
   \lBrack  A  \rBrack ^ \sharp  =  \lBrack   \delta (\hat{s})  \rBrack _A, \and
   \lBrack  A  \rBrack ^ \ell  =  \lBrack   \lambda ( \ell )  \rBrack _A \text{ for }  \ell   \in  L.
 \end{mathpar}
\end{definition}

\begin{example}
Let the CF-GKAT automaton from \Cref{example:cf-gkat-automaton} be \(A\), we find that
\begin{mathpar}
   \lBrack  A  \rBrack ^ \sharp _1 = \{ \emptyset \cdot \jmp{(\ell, 1)}, \{t\} p \emptyset \cdot \acc{2}, \{t\} p \{t\} \cdot \acc{2} \}
  \and
   \lBrack  A  \rBrack ^ \sharp _2 = \{ \emptyset \cdot \jmp{(\ell, 1)}, \{t\} \cdot \jmp{(\ell, 1)} \}
  \and
   \lBrack  A  \rBrack ^ \ell _1 =  \{ \alpha p \beta \cdot \acc{1} : \alpha, \beta \in \At \}
  \and
   \lBrack  A  \rBrack ^ \ell _2 = \{ \alpha p \beta \cdot \acc{2} : \alpha, \beta \in \At \}
\end{mathpar}
\end{example}

\subsection{Lowering CF-GKAT Automata to GKAT Automata}\label{sec:lowering-cf-gkat-automata-to-gkat}

The process to lower a CF-GKAT automaton $ \langle S,  \delta , \hat{s},  \lambda  \rangle $ into a GKAT automaton consists of two different components.
We first ``embed'' the indicator values into the state set; the new state set then becomes $S  \times  I$.
After that, we resolve all continuations in transitions using the jump map.

The second step requires some care.
When $ \delta (s, i,  \alpha ) = \jmp{( \ell , j)}$, the $ \alpha $-transition leaving the state $(s, i)$ will take the behavior of \( \ell \) starting from indicator \(j\) and atom \( \alpha \), that is \( \lambda ( \ell , j,  \alpha )\).
This by itself is alright, but we also need to account for subprograms such as $\comLabel{l};\; x := k;\; \comGoto{\ell'}$, which entails that $ \lambda ( \ell , j,  \alpha )$ points to a different label by returning $\jmp{( \ell ', k)}$.
%FIXME: many \jmp( \ell , i) still have  \ell  and i in the wrong place.
In turn, \( \lambda ( \ell ', k,  \alpha )\) may also yield another jump, et cetera.
To resolve these chained jumps, we will need to iterate the jump map, and terminate when either the result is no longer a jump, or an infinite loop is detected.

\begin{definition}[iteration lifting]\label{def: iteration lifting}
  Given a function \(h: X  \to  X +  \bot  + E\), where \(X\) is a finite set and \( \bot  + E\) specifies the ``exit results'', we can iterate $h$ until some exit value is reached, or a loop is detected.
  Formally, we define $\iter(h)$ as the least function satisfying for all $x \in X$ that
  \begin{mathpar}
    \iter(h)(x) =
        \begin{cases}
        \iter(h)(h(x)) & h(x) \in X \\
        h(x) & h(x) \in E
        \end{cases}
  \end{mathpar}
  in the directed-complete partial order (DCPO) on functions from $X$ to $X + \bot + E$ where $f \leq g$ when for all $x \in X$, we have that $f(x) = \bot$ or $f(x) = g(x)$.
  This makes $\iter(h)$ the least fixed point of the Scott-continuous function on this DCPO that sends $f$ to
  \[
    x \mapsto
        \begin{cases}
        f(h(x)) & h(x) \in X \\
        h(x) & h(x) \in E + \bot
        \end{cases}
  \]
  By Kleene's fixed point theorem this uniquely defines $\iter(h)$.
\end{definition}

\begin{remark}
  When $X$ is finite, we can directly compute $\iter(h)$ by iterating $h$.
  More precisely, we can define a helper function $\iter'(h): 2^X  \to  X  \to   \bot  + E$ with an additional parameter, as follows:
  \begin{align*}
  \iter'(h) & (M)(x)  \triangleq  \begin{cases}
       \bot  & \text{if } x  \in  M  \\
      h(x) & \text{if } x  \notin  M \text{ and } h(x)  \in   \bot  + E \\
      \iter'(h)(M  \cup  \{x\})(h(x)) & \text{if } x  \notin  M \text{ and } h(x)  \in  X
    \end{cases};
  \end{align*}
  It is then not too hard to show that $\iter(h) = \iter'(h)(\emptyset)$.
  Intuitively, \(M\) keeps track of explored values in \(X\); if an input has already been explored, then \(\iter\) will return \( \bot \).
  On the other hand, if \(h(x)\) falls into the exit set \( \bot  + E\), then \(\iter(h)\) will exit and return \(h(x)\).
  Finally, if \(h(x)\) falls into \(X\), then \(\iter(h)\) will continue to iterate with \(h(x)\) as input, and mark \(x\) as explored.
  Note that \(\iter(h)\) is total because \(X\) is finite, so elements in \(X\) cannot be explored twice.
\end{remark}

For the sake of clarity, when defining a function \(h\) to iterate, we will write \(\contWith\) for the injection $X \to X + \bot + E$, so that $\contWith(x)$ indicates iteration continues with the value $x$; and $\exitWith$ for the injection \( \bot  + E  \to  X +  \bot  + E\), so that $\exitWith(e)$ indicates iteration stops with the value $e$.

For jump resolution, iteration will continue when the result of \( \lambda \) is a jump, but exit otherwise.

\begin{definition}[Jump resolution]
 Let $S$ be a finite set, and let $ \lambda : L  \to  G(S)$ be a jump function.
 We define the resolved jump map ${ \lambda \! \downarrow }: L  \to  G(S)$, as follows:
 \[
  { \lambda \! \downarrow } = \iter \left(
    ( \ell , i,  \alpha )  \mapsto  \begin{cases}
      \contWith ( \ell ', i',  \alpha ) &  \lambda ( \ell , i,  \alpha ) = \jmp{( \ell ', i')} \\
      \exitWith ( \lambda ( \ell , i,  \alpha )) & \text{otherwise}
    \end{cases}
  \right)
 \]
\end{definition}

Because of the way $\iter$ works, this means that a cycle of $\jmp\!$-continuations (without actions in between) resolves to $\bot$.
Indeed, GKAT treats unproductive infinite iterations in the \command{while} loops with the same strategy~\cite{smolka_GuardedKleeneAlgebra_2020}, identifying non-terminating behavior with behavior that rejects explicitly.

Note that \( \lambda \! \downarrow \) resolves all internal jumps, i.e., \( \lambda \! \downarrow (\ell, i, \alpha)\) is never a jump continuation, because the iteration map used to define \( \lambda \! \downarrow \) keeps iterating on values of this form.
We can then use the resolved jump map to tie together the loose ends indicated by jump continuations.
In the process, we can reroute continuations of the form $\brk i$ to rejection, as these should not occur at the top level of a well-formed program (internal $\brk\!$-continuations are handled in the next section).

\begin{definition}[Lowering]
 Given a CF-GKAT automaton \(A  \triangleq   \langle S,  \delta , \hat{s},  \lambda  \rangle \) and $i  \in  I$, we define the GKAT automaton \({\mathit{A}\! \downarrow _{i}}  \triangleq   \langle S  \times  I,  \delta \! \downarrow , (\hat{s}, i) \rangle \), where $ \delta \! \downarrow $ is defined in two steps:
 \begin{align*}
   \delta '((s, i),  \alpha ) &  \triangleq
    \begin{cases}
       \lambda \! \downarrow ( \ell , i,  \alpha ) &  \delta (s, i,  \alpha ) = \jmp{( \ell , i)} \\
       \delta (s, i,  \alpha ) & \text{otherwise}
    \end{cases}\\
   \delta \! \downarrow ((s, i),  \alpha ) &  \triangleq
  \begin{cases}
    % \mathrlap and \hphantom is used for alignment purpose
    % the \mathrlap is the displayed expression
    % and \hphantom contains the "longest expression" for alignment
    \mathrlap{ \bot }\hphantom{ \lambda \! \downarrow ( \ell , i,  \alpha )} &  \delta '(s, i,  \alpha ) = \brk{j} \\
    \top  & \text{$ \delta '(s, i,  \alpha ) = \ret$ or $\delta'(s, i, \alpha) = \acc{j}$} \\
    \delta (s, i,  \alpha ) & \text{otherwise}
  \end{cases}
 \end{align*}
\end{definition}
Note that $ \delta \! \downarrow $ is well-defined: for all $s  \in  S$, $i  \in  I$ and $ \alpha   \in  \At$, it holds that $ \delta \! \downarrow ((s, i),  \alpha )  \in   \bot  +  \top  +  \Sigma   \times  (S  \times  I)$, as expected for a GKAT automaton on state set $S  \times  I$.

\begin{example}
If $A$ is the automaton from \Cref{example:cf-gkat-automaton}, then $\mathit{A}\! \downarrow _1$ is the following GKAT automaton:
\[
    \begin{tikzpicture}
        \node (shat1) {$(\hat{s}, 1)$};
        \node[right=of shat1] (s2) {$(s, 2)$};
        \node[left=of shat1] (s1) {$(s, 1)$};
        \draw ($(shat1) + (0mm, 7mm)$) edge[-latex] (shat1);
        \draw (shat1) edge[-latex] node[above] {\footnotesize $\{t\} \mid p$} (s2);
        \draw (shat1) edge[-latex] node[above] {\footnotesize $\emptyset \mid p$} (s1);
        \node[right=5mm of s2] (s2acc) {\footnotesize $\emptyset, \{t\}$};
        \draw (s2) edge[double,double distance=1.5pt,-implies] (s2acc);
        \node[left=5mm of s1] (s1acc) {\footnotesize $\emptyset, \{t\}$};
        \draw (s1) edge[double,double distance=1.5pt,-implies] (s1acc);
    \end{tikzpicture}
\]
Here, we use similar graphical conventions as before; for instance, $(\hat{s}, 1) \xrightarrow{\{ t \} \mid p} (s, 2)$ means that $\delta\!\downarrow((\hat{s}, 1), \{ t \}) = (p, (s, 2))$, and $(s, 2) \implies \emptyset,\{t\}$ means that $\delta\!\downarrow((s, 2), \emptyset) = \delta\!\downarrow((s, 2), \{t\}) = \top$.

In this lowered automaton, the jump continuation originally reached from $\hat{s}$ for indicator value $1$ and $t$ false (i.e., atom $\emptyset$) has been resolved via the jump map $\lambda$ to continue in state $s$.
This automaton also holds two copies of the state $s$, one for each possible indicator value; in this case, these states happen to have the same behavior, but in general they may be different.
\end{example}

Having defined our lowering operation, we can state its correctness as follows.

\begin{theorem}[Correctness of lowering]\label{the:cf-gkat-automaton-lowering-correctness}
 Let \(A  \triangleq   \langle S,  \delta , \hat{s},  \lambda  \rangle \) be a CF-GKAT automaton.
 The translation from CF-GKAT automata to GKAT automata commutes with the semantic jump resolution operator, in the sense that for $i  \in  I$, it holds that $ \lBrack A\! \downarrow _i \rBrack  =  \lBrack A \rBrack \! \downarrow ^ \sharp _{i}$.
\end{theorem}
\begin{proof}[Proof sketch]
More generally, we can prove that for any state $s$ of $A$, it holds that $ \lBrack (s, i) \rBrack _{A\! \downarrow _i} = G\! \downarrow ^ \sharp _{i}$, in which $G$ is the labeled family given by $G^\sharp =  \lBrack \delta(s) \rBrack _A$ and $G^\ell =  \lBrack \lambda(\ell) \rBrack _A$.
This property, which implies the main claim, can be proved by induction on the length of guarded words.
\end{proof}

\subsection{Converting Expressions to CF-GKAT Automata}\label{sec:thompson-construction}

The final piece of our puzzle is to convert CF-GKAT expressions to CF-GKAT automata.
To accomplish this, we generalize a construction proposed for GKAT, which turns a GKAT expression into a GKAT automaton in a trace-equivalent manner~\cite{smolka_GuardedKleeneAlgebra_2020}.
This construction, which was inspired by Thompson's construction to obtain a non-deterministic finite automaton from a regular expression~\cite{thompson_ProgrammingTechniquesRegular_1968}, proceeds by induction on the structure of the expression.
In contrast to the original, however, the Thompson construction for GKAT produces a GKAT automaton with a \emph{start dynamics}, also called an \emph{initial pseudostate}~\cite{smolka_GuardedKleeneAlgebra_2020}, instead of an explicit start state.
We adopt this shift in presentation to efficiently compose automata, avoiding the silent transitions in the original~\cite{thompson_ProgrammingTechniquesRegular_1968}.

\begin{definition}
 A CF-GKAT automaton with start dynamics \(A  \triangleq   \langle S,  \delta ,  \iota ,  \lambda  \rangle \) consists of $S$, $ \delta $ and $ \lambda $ as in a CF-GKAT automaton, in addition to a start dynamics \( \iota   \in  G(S)\).
 The labeled family defined by a CF-GKAT automaton with start dynamics, also denoted $ \lBrack A \rBrack $, is simply given by $ \lBrack A \rBrack ^\sharp =  \lBrack \iota \rBrack _A$ and $ \lBrack  A  \rBrack ^ \ell  =  \lBrack   \lambda ( \ell )  \rBrack _A$, where the right-hand sides are defined as for plain CF-GKAT automata.
\end{definition}

It should be clear that CF-GKAT automata with start dynamics can easily be converted to plain CF-GKAT automata by adding a start state \(\hat{s}\) that takes the behavior of the start dynamics \( \iota \):
\begin{equation}\label{cons: CF-GKAT pseudo start to CF-GKAT automata}
  \langle S,  \delta ,  \iota ,  \lambda  \rangle   \mapsto   \langle S + \hat{s},  \delta _ \iota , \hat{s},  \lambda  \rangle ,
 \qquad
 \text{where}
 \qquad
  \delta _ \iota (s, i,  \alpha )  \triangleq
 \begin{cases}
   \iota (i,  \alpha )    & \text{if } s = \hat{s} \\
   \delta (s, i,  \alpha ) & \text{if } s  \neq  \hat{s}
 \end{cases}
\end{equation}

Our construction turns a CF-GKAT expression \(e\) into a CF-GKAT automaton with start dynamics, which we call the \emph{Thompson automaton for \(e\)}.
The following paragraphs describe the construction and intuition behind each case in the construction.
We make the simplifying assumption that each $\ell$ appears at most once in a subexpression of the form $\comLabel{\ell}$.
Moreover, we write \({!}\) for the unique function \({!}:  \emptyset   \to  X\) (for any set $X$), and in inductive cases we denote the Thompson automata for \(e_{1}\) and \(e_{2}\) by \(A_{1}  \triangleq   \langle S_{1},  \delta _{1},  \iota _{1},  \lambda _{1} \rangle \) and \(A_{2}  \triangleq   \langle S_{2},  \delta _{2},  \iota _{2},  \lambda _{2} \rangle \) respectively.
Finally, to compress our notation, we will use the compact syntax of GKAT~\cite{smolka_GuardedKleeneAlgebra_2020} for \command{if}-statements and \command{while}-loops:
\begin{mathpar}
  e_{1} +_b e_{2}  \triangleq  \comITE{b}{e_{1}}{e_{2}}, \and
  e_{1}^{(b)}  \triangleq  \comWhile{b}{e_{1}}.
\end{mathpar}

\subsubsection*{Converting \command{break}, \command{return}, \command{goto}, and indicator assignment:}
recall that the semantics of \(\comBrk\), \(\comRet\), \(\comGoto{ \ell }\), and indicator assignments simply emit the corresponding continuations.
Thus, the non-trivial part of these Thompson automata are the start dynamics \( \iota \), which output said continuations.
Because the start dynamics does not need to transition anywhere, there are no further states; because these programs also do not contain any $\comLabel\!$ primitives, there is no need to assign a dynamics for them either.
We thus construct the following automata:
\begin{align*}
  S_{\comBrk} &  \triangleq   \emptyset  &
     \delta _{\comBrk} &  \triangleq  {!} &
     \iota _{\comBrk}(i,  \alpha ) &  \triangleq  \brk{i} &
     \lambda _{\comBrk}( \ell , i, \alpha) &  \triangleq   \bot  \\
  S_{\comRet} &  \triangleq   \emptyset  &
     \delta _{\comRet} &  \triangleq  {!} &
     \iota _{\comRet}(i,  \alpha ) &  \triangleq  \ret &
     \lambda _{\comRet}( \ell , i, \alpha) &  \triangleq   \bot  \\
  S_{\comGoto{ \ell }} &  \triangleq   \emptyset  &
     \delta _{\comGoto{ \ell }} &  \triangleq  {!} &
     \iota _{\comGoto{ \ell }}(i,  \alpha ) &  \triangleq  \jmp{( \ell , i)} &
     \lambda _{\comGoto{ \ell }}( \ell , i, \alpha) &  \triangleq   \bot  \\
  S_{x := i} &  \triangleq   \emptyset  &
     \delta _{x := i} &  \triangleq  {!} &
     \iota _{x := i}(i,  \alpha ) &  \triangleq  \acc{i} &
     \lambda _{x := i}( \ell , i, \alpha) &  \triangleq   \bot
\end{align*}

\subsubsection*{Converting tests and primitive actions:}
the conversions of primitive tests and primitive actions largely inherit the Thompson construction for GKAT\@.
The Thompson automaton for tests \(\comAssert{b}\) contains only a start dynamics, which accepts the input indicator-atom pairs if and only if they satisfy \(b\).
The Thompson's automata for primitive actions \(p\) contains a start dynamics that always executes the action \(p\) before transitioning to the unique state, which accepts unconditionally.
\begin{align*}
  S_{\comAssert{b}} &  \triangleq   \emptyset  &
    S_{p} &  \triangleq  \{*\} \\
   \delta _{\comAssert{b}} &  \triangleq  {!} &
     \delta _{p}(s, i,  \alpha ) &  \triangleq  \acc{i} \\
   \iota _{\comAssert{b}}(i,  \alpha ) &  \triangleq  \begin{cases}
      \acc{i} & (i,  \alpha )  \in   \lBrack b \rBrack  \\
       \bot  & (i,  \alpha )  \notin   \lBrack b \rBrack
    \end{cases} &
     \iota _{p}(i,  \alpha ) &  \triangleq  (p, *, i) \\
   \lambda _{\comAssert{b}}( \ell , i, \alpha) &  \triangleq   \bot  &
     \lambda _{p}( \ell , i, \alpha) &  \triangleq   \bot
\end{align*}

\subsubsection*{Converting labels:}
recall that \(\comLabel{ \ell '}\) is a does not affect the semantics from the start of the program, i.e., should be the same as the sequential identity \(\comAssert{\true}\).
However, the behavior of \(\comLabel{ \ell '}\) and \(\comAssert{\true}\) diverges when we consider the jump map: if $\ell = \ell'$, then executing $\comLabel{\ell'}$ starting from $\ell$ lets us reach the end of the program, whereas $\comAssert{\true}$ does not contain any label.
We thus end up with the following Thompson automaton for $\comLabel{\ell'}$:
\begin{align*}
  S_{\comLabel{ \ell '}} &  \triangleq   \emptyset  &
   \delta _{\comLabel{ \ell '}} &  \triangleq  {!} &
   \iota _{\comLabel{ \ell '}}(i,  \alpha ) &  \triangleq  \acc{i} &
   \lambda _{\comLabel{ \ell '}}( \ell , i,  \alpha ) &  \triangleq  \begin{cases}
    \acc{i} &  \ell  =  \ell '  \\
     \bot  &  \ell   \neq   \ell '
  \end{cases}
\end{align*}

% FIXME: I think this is too on the nose... Basically just reiterating the definition
% TK: I think it's unavoidable at this point
\subsubsection*{Converting \command{if} statements:}
the Thompson automaton for \(\comITE{b}{e_{1}}{e_{2}}\) is also similar to that of GKAT\@: if the input indicator-atom pair satisfies (resp.\ falsifies) \(b\), then \(\iota\) will enter the Thompson automaton of \(e_{1}\) (resp.\ $e_2$) by taking on the behavior of \( \iota _{1}\) (resp.\ $\iota_2$).
The jump map \( \lambda \) assigns the entry point for label \( \ell \) based on where it appears: if it appears in \(e_{1}\), then \( \ell \) will take its entry point in \(A_{1}\); similarly, if \( \comLabel{\ell} \) appears in \(e_{2}\), \( \ell \) will take its entry point in \(A_{2}\).
\begin{minipage}{0.54\textwidth}
  \begin{align*}
  S_{e_{1} +_b e_{2}} &  \triangleq  S_{1} + S_{2} \\
   \lambda _{e_{1} +_b e_{2}}( \ell ) &  \triangleq  \begin{cases}
     \lambda _{1}( \ell ) & \comLabel{ \ell } \text{ appears in } e_{1}\\
     \lambda _{2}( \ell ) & \comLabel{ \ell } \text{ appears in } e_{2}\\
     \bot  & \text{otherwise}
  \end{cases}
  \end{align*}
\end{minipage}
\begin{minipage}{0.44\textwidth}
\begin{align*}
   \delta _{e_{1} +_b e_{2}}(s) &  \triangleq  \begin{cases}
     \delta _{1}(s) & \text{if } s  \in  S_{1} \\
     \delta _{2}(s) & \text{if } s  \in  S_{2} \\
  \end{cases} \\
   \iota _{e_{1} +_b e_{2}}(i,  \alpha ) &  \triangleq  \begin{cases}
     \iota _{1}(i,  \alpha ) & (i,  \alpha )  \in   \lBrack b \rBrack  \\
     \iota _{2}(i,  \alpha ) & (i,  \alpha )  \notin   \lBrack b \rBrack
  \end{cases}
\end{align*}
\end{minipage}

\subsubsection*{Converting Sequencing:}
Sequencing of automata can be defined by \emph{uniform continuations}~\cite{smolka_GuardedKleeneAlgebra_2020}, which combine two dynamics $h_{1}, h_{2}  \in  G(S)$ into a new dynamics \(h_{1}[h_{2}]\).
The latter acts like $h_{1}$ in almost all cases, except when $h_{1}$ accepts --- then it will take on the behavior of \(h_{2}\).
In other words, \(h_{1}[h_{2}]\) connects all the accepting transition of \(h_{1}\) to \(h_{2}\).
Uniform continuation is typically used to compose two automata or add self-loops to an automaton; it can be formally defined as follows.
\begin{definition}[Uniform Continuation]
  Let $S$ be a set.
  Given two dynamics $h_{1}, h_{2}  \in  G(S)$, their \emph{uniform continuation} is the dynamics $h_{1}[h_{2}] \in  G(S)$, defined as follows:
  \[
    h_{1}[h_{2}](i,  \alpha )  \triangleq
    \begin{cases}
    h_{2}(i',  \alpha ) & \text{if } h_{1}(i,  \alpha ) = \acc{i'} \\
    h_{1}(i,  \alpha )  & \text{otherwise}
    \end{cases}
  \]
\end{definition}
To construct the Thompson automaton for \(e_{1}; e_{2}\), we will simply connect all the accepting transitions in \(A_{1}\) to \(A_{2}\) by applying uniform continuations on start dynamics \( \iota _{1}\), transitions \( \delta _{1}\), and jump map \( \lambda _{1}\), while preserving the dynamics in \(A_{2}\).
Formally, this works out as follows:
\begingroup%
\allowdisplaybreaks%
\begin{align*}
  S_{e_{1}; e_{2}} &  \triangleq  S_{1} + S_{2} &
   \iota _{e_{1}; e_{2}} &  \triangleq   \iota _{1}[ \iota _{2}] \\
   \delta _{e_{1}; e_{2}}(s) &  \triangleq  \begin{cases}
     \delta _{1}(s)[ \iota _{2}] & \text{if } s  \in  S_{1} \\
     \delta _{2}(s) & \text{if } s  \in  S_{2} \\
  \end{cases} &
   \lambda _{e_{1}; e_{2}}( \ell , i,  \alpha ) &  \triangleq  \begin{cases}
     \lambda _{1}( \ell )[\iota_2] & \comLabel{ \ell } \text{ appears in } e_{1}\\
     \lambda _{2}( \ell ) & \comLabel{ \ell } \text{ appears in } e_{2}\\
     \bot  & \text{otherwise}
  \end{cases}
\end{align*}
\endgroup%

\subsubsection*{Converting \command{while} loops:}
Like GKAT automata, CF-GKAT automata require every transition between states to execute a primitive action.
This presents a unique challenge in defining the start dynamics for while loops.
Namely, a loop may not immediately encounter a primitive action on its first iteration, but this may trigger a change in indicator value that causes a primitive action to be executed in the \emph{next} iteration, or the one after that, et cetera.
This is reminiscent of jump resolution, where we also have to chase through some indirection to arrive at the right transition.

As a concrete example of this phenomenon, consider the following CF-GKAT program:
\begin{equation}%
\label[prog]{prog:loop-head-iter-example}
  \begin{aligned}
    \comWhile{\true}{\{
      & \comITE{x = 0}{x := 1 \\
      &}{\comITE{x = 1}{\comBrk
      }{\{\; \comAssert{\true} \;\}}} \;\}}
  \end{aligned}
 \end{equation}
 If this program starts with the indicator $0$, then the first continuation (\(\brk{1}\)) will be encountered on the second iteration of the loop.
Even worse, when starting with an indicator value like \(x = 2\), the program will enter an infinite loop and never encounter a primitive action or continuation.

Fortunately, these difficulties can be resolved by the \(\iter\) function (\Cref{def: iteration lifting}): we repeat the start of the loop body until we encounter some productive behavior, or we get stuck.
\begin{definition}[Iterated Start Dynamics]
  Let $S$ be a set, let $h  \in  G(S)$, and $b  \in  \BExp$.
  We can use the \(\iter\) function to define $h^{b}: G(S)$, as follows:
  \[
   h^b  \triangleq
   \iter\left(
     (i,  \alpha )  \mapsto
     \begin{cases}
       \exitWith(\acc{i}) & \text{if } (i,  \alpha )  \notin   \lBrack b \rBrack  \\
       \contWith(i',  \alpha ) & \text{if } (i,  \alpha )  \in   \lBrack b \rBrack  \text{ and } h(i,  \alpha ) = \acc{i'} \\
       \exitWith(h(i,  \alpha )) & \text{otherwise}
     \end{cases}
   \right)
  \]
\end{definition}
In the first case, the input \((i, a)\) does not satisfy \(b\), causing the while loop to terminate.
In the second case, the loop body accepts \((i,  \alpha )\) immediately and returns the exit indicator value \(i'\), thus the iteration of loop body will continue with \((i',  \alpha )\).
In the final case, the program executes an action or encounters a non-local control, which is when the iteration can be stopped.
\begin{example}[Iterated Start Dynamics]
  Consider~\cref{prog:loop-head-iter-example} above with indicator set \(\{0, 1, 2\}\), no primitive action, no label, and no primitive test.
  Then the only atom is \( \emptyset \), and the Thompson automaton \(A_{1}  \triangleq   \langle S_{1},  \delta _{1},  \iota _{1},  \lambda _{1} \rangle \) for the loop body can be computed as follows:
  \begin{align*}
    &&  \iota _{1}(0,  \emptyset ) &  \triangleq  \acc{1} &&&\\
    S_{1} &  \triangleq   \emptyset  &
     \iota _{1}(1,  \emptyset ) &  \triangleq  \brk{1} &
     \delta _{1} &  \triangleq  {!} &
     \lambda (\ell, i, \alpha) &  \triangleq  \bot \\
    &&  \iota _{1}(2,  \emptyset ) &  \triangleq  \acc{2} &&&
  \end{align*}
  Then we compute the iterated start dynamics \( \iota ^{\true}\) with input \((0, \emptyset )\) and \((2, \emptyset )\):
  \begin{align*}
    { \iota _{1}}^{\true}(0,  \emptyset )
    & = { \iota _{1}}^{\true}(1,  \emptyset )
      & \text{because }(0,  \emptyset )  \in   \lBrack \true \rBrack  \text{ and } { \iota _{1}}(0,  \emptyset ) = \acc{1} \\
    & = \brk{1}
      & \text{because } { \iota _{1}}(1,  \emptyset ) = \brk{1} \\[5px]
    { \iota _{1}}^{\true}(2,  \emptyset )
    & = { \iota _{1}}^{\true}(2,  \emptyset )
      & \text{because }(2,  \emptyset )  \in   \lBrack \true \rBrack  \text{ and }  \iota _{1}(2,  \emptyset ) = \acc{2} \\
    & =  \bot
      & \text{because the input \((2,  \emptyset )\) is already explored}
  \end{align*}
\end{example}

With the start dynamics defined, we still need to resolve structures within the loop body, like the \(\comBrk\)-continuation.
To perform $\comBrk$-resolution, we extend the \( \lfloor - \rfloor \) operator to dynamics.
\begin{definition}
Let $S$ be a set, and let $h  \in  G(S)$.
We define $ \lfloor h \rfloor   \in  G(S)$ by lifting \(h\) via \( \lfloor - \rfloor \) (c.f. \Cref{def:intermediate-sequencing-loops}) when it returns a continuation, that is to say:
\[
   \lfloor h \rfloor (i,  \alpha ) = \begin{cases}
   \lfloor h(i,  \alpha ) \rfloor  & \text{if } h(i,  \alpha )  \in  \mathcal{C} \\
  h(i,  \alpha )   & \text{otherwise}
  \end{cases}
\]
\end{definition}

Finally, the transition function \( \delta \) and jump map \( \lambda \) can be defined by first connecting \( \delta _{1}\) and \( \lambda _{1}\) back to the start dynamics \( \iota^{(b)} \), forming a loop in the automaton, and then resolving the \(\brk{i}\) continuations using \( \lfloor - \rfloor \).
Formally, the Thompson automaton for loop \(e^{(b)}\) is defined as follows:
\begin{align*}
  S_{e_{1}^{(b)}} &  \triangleq  S_{1} &
   \delta _{e_{1}^{(b)}} &  \triangleq   \lfloor   \delta _{1}(s)[{ \iota _{1}}^{b}]  \rfloor  &
   \iota _{e_{1}^{(b)}} &  \triangleq   \lfloor  { \iota _{1}}^{b}  \rfloor  &
   \lambda _{e_{1}^{(b)}}( \ell ) &  \triangleq   \lfloor   \lambda ( \ell )[{ \iota _{1}}^{b}]  \rfloor
\end{align*}
Note how in the case of the jump map, we resolve $\brk\!$-continuations \emph{after} the uniform continuation with the iterated start dynamics.
This way, a jump to a label inside the loop that hits a $\comBrk$ statement immediately will continue executing execution after the loop, as is to be expected.

% FIXME: I think this is good to say, but should be in a separate remark after the construction is fully presented - T
% In the worst case, it is possible for \(I'\)
% to exhaust all of \(I\) before an infinite loop is found,
% which means computing \( \gamma \) can take \(|I|\) time
% for every indicate \(i\) and atoms \( \alpha \).
% However, for a fixed atom \( \alpha \), we can cache the result of each \(i\),
% leading to a \(|I|\)-timed algorithm to compute \( \gamma \) for every input \(i\).

The correctness of the Thompson construction for CF-GKAT can be formulated as follows:
\begin{theorem}[Thompson's construction preserves the continuation semantics]\label{the:thompson-correctness}
  Given an expression $e  \in  \Exp$, let $A_e$ be the Thompson automaton for $e$.
  Now \( \lBrack e \rBrack  =  \lBrack A_e \rBrack \).
 \end{theorem}
 \begin{proof}[Proof sketch]
 By induction on $e$.
 This proof is somewhat tedious, but ultimately doable.
 The inductive cases are very similar to the ones for the Thompson construction in plain GKAT~\cite{smolka_GuardedKleeneAlgebra_2020}; the only difference is that, this time, the semantics of the jump map needs to be taken into account.
 \end{proof}

\subsection{Algorithm, Completeness, and Complexity}

With the definitions of lowering and Thompson's construction established, the decision procedure mostly follows.
Nevertheless, it remains essential to define the alphabets $\Sigma$, $T$, $I$, and $L$, representing the set of primitive actions, primitive tests, indicator values, and labels, respectively.
We may safely restrict primitive actions, primitive tests, and labels to those explicitly present in the expression, as expanding the alphabet beyond these will preserve (the equational theory of) the trace semantics.

The set of indicator values, however, occupies a unique position.
If the initial indicator value is absent from the program, the program's traces may diverge from traces starting from the present indicator values.
\Cref{prog:loop-head-iter-example} is one of the witnesses of this phenomenon: if the initial indicator value is 0 or 1, the program terminates; however, when starting from an indicator value that does not appear in the program, the program will loop indefinitely.
An even simpler example is
\(\comAssert (x = 1); p\),
which executes \(p\) if the initial indicator value is 1, but rejects when started with indicator values that are not present in the program.
Fortunately, given an expression \(e\), it is not hard to show that if neither \(i\) nor \(i'\) appears in \(e\), then the behaviors for both values coincide, i.e.,
\( \forall   \ell ,  \lBrack e \rBrack _{i'}^{ \ell } =  \lBrack e \rBrack _{i}^{ \ell }\).
Thus, when gathering the indicator values, it suffices to take those that appear explicitly in the program augmented with a fresh value \(*\) that does not appear in the program.

We summarize our decision procedure as follows:
\begin{enumerate}
  \item
  Given two CF-GKAT programs \(e, f\), we first collect their alphabets $\Sigma$, $T$, $I$, and $L$.
  We gather the sets of primitive actions \(\Sigma\), primitive tests \(T\), and labels \(L\) that are present in either \(e\) or \(f\).
  Additionally, we identify the set of indicator values \(I\), encompassing those found in either \(e\) or \(f\), along with an additional indicator \(*\) that is exclusive to both programs.
  \item
  We then proceed to compute the Thompson automata of \(e\) and \(f\) and convert them into CF-GKAT automata, denoted as \(A_e\) and \(A_f\).
  It is noteworthy that these automata preserve the continuation semantics (\Cref{the:thompson-correctness}).
  Formally,
  \begin{mathpar}
     \lBrack A_e \rBrack  =  \lBrack e \rBrack, \and  \lBrack A_f \rBrack  =  \lBrack f \rBrack .
  \end{mathpar}
  \item
  Subsequently, we lower both \(A_e\) and \(A_f\) to GKAT automata \(A_e\! \downarrow _{i}\) and \(A_f\! \downarrow _{i}\) for each \(i  \in  I\).
  By \Cref{the:cf-gkat-automaton-lowering-correctness}, the latter exhibits the same traces as \(e\) and \(f\) starting from \(i\):
  \begin{mathpar}
     \lBrack A_e \! \downarrow _{i} \rBrack  =  {\lBrack A_e \rBrack \! \downarrow^\sharp _{i}} =  {\lBrack e \rBrack \! \downarrow^\sharp _{i}}, \and
     \lBrack A_f \! \downarrow _{i} \rBrack  =  {\lBrack A_f \rBrack \! \downarrow^\sharp _{i}} =  {\lBrack f \rBrack \! \downarrow^\sharp _{i}}
  \end{mathpar}
  \item
  Finally, we run a GKAT automata equivalence check~\cite{smolka_GuardedKleeneAlgebra_2020} on \(A_e \! \downarrow _{i}\) and \(A_f \! \downarrow _{i}\) for each \(i  \in  I\), and return true when all GKAT automata equivalence checks return true.
\end{enumerate}

\paragraph*{Soundness and completeness:}
The soundness and completeness of this algorithm now follow as a corollary of the corresponding properties for the decision procedure in GKAT\@.
We denote the algorithm introduced above as \(\mathrm{equiv}_{\mathsf{CFGKAT}}\), while the decision algorithm for equivalence of GKAT automata is denoted as \(\mathrm{equiv}_{\mathsf{GKAT}}\).
Thus, we establish the equivalence:
\begin{align*}
  \mathrm{equiv}_{\mathsf{CFGKAT}}(e, f)
  &  \iff   \forall  i  \in  I,\ \mathrm{equiv}_{\mathsf{GKAT}}(A_e \! \downarrow _{i}, A_f \! \downarrow _{i}) \tag{def. $\mathrm{equiv}_{\mathsf{CFGKAT}}$} \\
  &  \iff   \forall  i  \in  I,\  \lBrack A_e \! \downarrow _{i} \rBrack  =  \lBrack A_f \! \downarrow _{i} \rBrack  \tag{corr. $\mathrm{equiv}_{\mathsf{GKAT}}$} \\
  &  \iff   {\lBrack A_e \rBrack \! \downarrow^\sharp}  =  {\lBrack A_f \rBrack  \! \downarrow^\sharp} \tag{\Cref{the:cf-gkat-automaton-lowering-correctness}} \\
  &  \iff   {\lBrack e \rBrack \! \downarrow^\sharp}  = {\lBrack f \rBrack \! \downarrow^\sharp} \tag{\Cref{the:thompson-correctness}}
\end{align*}
Therefore \(e\) and \(f\) are trace equivalent if and only if \(\mathrm{equiv}_{\mathsf{GKAT}}(e, f)\) returns true.

\paragraph*{Algorithm complexity:}
We now give a rough account of the computational cost for deciding equivalence of CF-GKAT programs.
Like GKAT, we consider $T$ to be fixed for the purpose of analyzing complexity; otherwise, deciding equivalence of (CF-)GKAT programs is co-NP hard~\cite{smolka_GuardedKleeneAlgebra_2020}. % chktex 36

Starting with an expression \(e\), we observe that the number of states in the Thompson automaton \(A_e\) is bounded by \(|e|\), which is the size of $e$ as a term.
While computing this automaton we must bear in mind that deriving the iterated start dynamics of a loop using $\iter$ may take up to $|I|$ steps (assuming we use memoization to compute the output for each input), and this happens at most $|e|$ times.
After lowering, each GKAT automaton \(A_e\! \downarrow \) contains at most \(|I| \times |e|\) states, and computing the jump resolution using $\iter$ takes on the order of $|L| \times |I|$ steps (again using memoization).

For each \(i  \in  I\), determining the equivalence between \(A_e\! \downarrow _{i}\) and \(A_f\! \downarrow _{i}\) is accomplished in (nearly\footnote{We omit the factor coming from the inverse Ackermann function $\hat{\alpha}(n)$, which is at most $5$ for any realistic number of states, out of consideration for the sake of simplicity.}) linear time relative to the number of states, so this check takes about \(|I| \times (|e|+|f|)\) steps.
To verify trace equivalence between \(e\) and \(f\), this check is required for \(A_e\! \downarrow _{i}\) and \(A_f\! \downarrow _{i}\) across all \(i \in I\).

Therefore, the overall time spent on the equivalence checks is on the order of \(|I|^2  \times  (|e|+|f|)\), while computing the automata can take roughly $|L| \times |I| + |I| \times (|e| + |f|)$ time.
This implies that the algorithm's complexity scales (nearly) linearly with the sizes of \(e\) and \(f\), linearly in the size of $L$, and quadratically with respect to the number of indicator values \(|I|\).

\section{Control Flow Validation}%
\label{sec:experiments}

% Our theory can be specialized into control flow verification in two different way: first, we can decide the equivalence between a program with a flowchart, where the flowchart is represented by a GKAT automaton; second, we can directly decide the equivalence between to programs.

We hypothesize that CF-GKAT can be a useful tool to check whether two programs have the same control flow --- that is, under the same circumstances, they execute the same sequence of actions.
In this section, we report on the performance of CF-GKAT in two particular use cases.

The first use case concerns the \emph{control flow structuring} stage of a decompiler.
Briefly put, a decompiler is a program that infers a high-level language representation of a binary executable file.
In earlier stages, the decompiler builds a \emph{control-flow graph} from the binary~\cite{DBLP:conf/pldi/VerbeekBFR22}, which represents how control transfers from one part to the next.
The control flow structuring pass is tasked with inferring an equivalent representation of this control flow in terms of constructs like \command{if-then-else} and \command{while-do}.
The second use case is about refactoring operations aimed at making code more readable by eliminating $\comGoto$ instructions, which may introduce indicator variables~\cite{erosa-hendren-1994,DBLP:journals/tsi/CasseFRS02}.

CF-GKAT can be used to validate that the control flow of the input (graph or code) corresponds to that of the output code.
In general, such algorithms can never be \emph{complete} with respect to input-output equivalences, as they cannot automatically validate the correctness of refactorings depend on the meaning of primitive commands.
However, CF-GKAT should be applicable to refactoring operations that rearrange the code for the sake of improving the presentation of the control flow.

\subsection{Implementation}

We target C, as it is a widely used programming language that allows for non-local control flow.
Our implementation is written in OCaml and is divided into a front-end and a back-end~\cite{artifact}.

The \emph{front-end} converts a function defined in a C file to a CF-GKAT expression,
using the \command{clangML}\footnote{\url{https://memcad.gitlabpages.inria.fr/clangml/}} OCaml bindings for \command{clang}.\footnote{\url{https://clang.llvm.org/}}
The conversion first determines which variables, if any, qualify as indicator variables and picks (at most) one.
Next, it lifts the C syntax tree to a CF-GKAT program, mapping (1)~all assignments and comparisons of the indicator variable to their CF-GKAT counterparts, (2)~all other primitive statements to uninterpreted actions, and (3)~all control flow constructs to the corresponding CF-GKAT operator whenever possible.
% TODO: We should probably have a figure with an example somewhere here.

The \emph{back-end} is responsible for compiling a CF-GKAT expression to a GKAT automaton, and for comparing two GKAT automata.
This GKAT automaton is derived using the construction from \Cref{section:decision procedure}; equivalence is decided by checking bisimulation~\cite{smolka_GuardedKleeneAlgebra_2020}.

These parts combine into a tool that accepts two C files, and checks whether their functions (paired by name) have the same control flow.
Not all C programs can be converted to CF-KAT expressions, and so our front-end is not complete; instead, we aimed to support a sizeable portion of the code in GNU Coreutils to perform the experiments described in the next section.

\begin{remark}
In particular, our front-end internally converts \command{for}-loops to \command{while}-loops, and \command{do-while} loops to (partially unrolled) \command{while} loops. The former transformation is sound; the latter is admissible when the loop body does not contain \command{break} or a label, which is the case for our experiment below.
\end{remark}

\subsection{A Case Study from GNU Coreutils}

% We tested our implementation by attempting to validate the results of certain program transformations.
% In each experiment we take a function body, apply a code transformation, and then check whether our implementation can certify that the input code has the same control flow as the output code.

We used GNU Coreutils as a source of C code containing non-local control flow.
As is customary with projects of this size, the code is organized into several files that collectively define hundreds of functions, ranging from very simple to very complex.
Throughout this section, we focus on the function \command{mp\_factor\_using\_pollard\_rho} in \command{factor.c}.
The code in Figure~\ref{c:orig} shows an abridged version of this function in Coreutils 9.5.
We now discuss the two transformations we targeted.

\begin{figure}[hbtp]
\centering
 \begin{subfigure}{0.46\textwidth}
 \begin{lstlisting}[basicstyle=\tiny\ttfamily]
  static void mp_factor_using_pollard_rho(...) {
  mpz_t x, z, y, P;
  mpz_t t, t2;
  devmsg("...", a);
  ...
  while (mpz_cmp_ui(n, 1) != 0) {
    for (;;) {
      do {
        mpz_mul(t, x, x);
        ...
        if (k % 32 == 1) {
          mpz_gcd(t, P, n);
          if (mpz_cmp_ui(t, 1) != 0)
            goto factor_found;
          mpz_set(y, x);
        }
      } while (--k != 0);
      mpz_set(z, x);
      ...
      for (unsigned long long int i = 0; i < k; i++) {
        mpz_mul(t, x, x);
        ...
      }
      mpz_set(y, x);
    }
  factor_found:
    do {
      mpz_mul(t, y, y);
      ...
    } while (mpz_cmp_ui(t, 1) == 0);
    mpz_divexact(n, n, t);
    if (!mp_prime_p(t)) {
      devmsg("...");
      ...
    } else {
      mp_factor_insert(factors, t);
    }
    if (mp_prime_p(n)) {
      mp_factor_insert(factors, n);
      break;
    }
    mpz_mod(x, x, n);
    ...
  }
  mpz_clears(P, t2, t, z, x, y, nullptr);
}
 \end{lstlisting}
 \caption{\label{c:orig}Original code of the function.}
 \end{subfigure}
 \begin{subfigure}{0.46\textwidth}
 \begin{lstlisting}[basicstyle=\tiny\ttfamily]
void mp_factor_using_pollard_rho() {
  pact(197);
  ...
  do {
    if (pbool(83)) {
      pact(194);
    }

  } while (0);
  pact(193);
  ...
  while (pbool(83)) {
    for (;;) {
      do {
        pact(176);
        ...
        if (pbool(183)) {
          pact(182);
          if (pbool(83)) {
            goto factor_found;
          }
          pact(173);
        }

      } while (pbool(181));
      pact(180);
      ...
      for (pact(177); pbool(138); pbool(59)) {
        pact(176);
        ...
      }
      pact(173);
    }
  factor_found:
    do {
      pact(172);
      ...
    } while (pbool(83));
    pact(166);
    if (!pbool(165)) {
      do {
        if (pbool(83)) {
          pact(164);
        }
      } while (0);
      pact(163);
    } else {
      pact(161);
    }
    if (pbool(160)) {
      pact(159);
      break;
    }
    pact(158);
    ...
  }
  pact(155);
}
 \end{lstlisting}
 \caption{\label{c:blinded}``Blinded'' code of the function.}
 \end{subfigure}
 \caption{Different versions of \command{mp\_factor\_using\_pollard\_rho} in \command{factor.c}, part of GNU Coreutils.}
 \vspace{3mm} % HACK
\end{figure}
\begin{figure}
 \begin{subfigure}{0.46\textwidth}
 \begin{lstlisting}[basicstyle=\tiny\ttfamily]
void mp_factor_using_pollard_rho(void) {
  pact(0xc5);
  ...
  if (pbool(0x53)) {
    pact(0xc2);
  }
  pact(0xc1);
  ...
LAB_00100a0c:
  if (pbool(0x53)) {
  LAB_00100a2d:
    pact(0xb0);
    ...
    if (pbool(0xb7)) {
      pact(0xb6);
      if (pbool(0x53))
        goto LAB_00100b47;
      pact(0xad);
    }
    if (!pbool(0xb5)) {
      pact(0xb4);
      ...
      while (pbool(0x8a)) {
        pact(0xb0);
        ...
      }
      pact(0xad);
    }
    goto LAB_00100a2d;
  }
LAB_00100c34:
  pact(0x9b);
  return;
LAB_00100b47:
  do {
    pact(0xac);
    ...
  } while (pbool(0x53));
  pact(0xa6);
  if (!pbool(0xa5)) {
    if (pbool(0x53)) {
      pact(0xa4);
    }
    pact(0xa3);
  } else {
    pact(0xa1);
  }
  if (pbool(0xa0)) {
    pact(0x9f);
    goto LAB_00100c34;
  }
  pact(0x9e);
  ...
  goto LAB_00100a0c;
}
 \end{lstlisting}
 \caption{\label{c:ghidra}Ghidra's decompilation of the function.}
 \end{subfigure}
 \begin{subfigure}{0.46\textwidth}
 \begin{lstlisting}[basicstyle=\tiny\ttfamily]
void mp_factor_using_pollard_rho() {
  int factor_found = 0;
  pact(197);
  ...
  do {
    if (pbool(83)) {
      pact(194);
    }
  } while (0);
  pact(193);
  ...
  while (pbool(83)) {
    for (; factor_found == 0;) {
      do {
        pact(176);
        ...
        if (pbool(183)) {
          pact(182);
          if (pbool(83)) {
            factor_found = 1;
          }
          if (factor_found == 0)
            pact(173);
        }
      } while ((factor_found == 0) && pbool(181));
      if (factor_found == 0) {
        pact(180);
        ...
        for (pact(177); pbool(138); pbool(59)) {
          pact(176);
          ...
        }
        pact(173);
      }
    }
    factor_found = 0;
    do {
      pact(172);
      ...
    } while (pbool(83));
    pact(166);
    if (!pbool(165)) {
      do {
        if (pbool(83)) {
          pact(164);
        }
      } while (0);
      pact(163);
    } else {
      pact(161);
    }
    if (pbool(160)) {
      pact(159);
      break;
    }
    pact(158);
    ...
  }
  pact(155);
}
 \end{lstlisting}
 \caption{\label{c:calipso}The original, with \command{goto}s removed by Calipso.}
 \end{subfigure}
 \caption{Different versions of \command{mp\_factor\_using\_pollard\_rho} in \command{factor.c}, part of GNU Coreutils.}
\end{figure}

\paragraph{Compilation-decompilation}
We hypothesize that CF-GKAT should be usable to validate the output of control flow structuring algorithms in decompilers.
To fully test this hypothesis, we would need to have access to the internal representations used before and after control flow structuring, and convert those to GKAT automata and CF-GKAT expressions, respectively.
Unfortunately, doing this would entail a substantial engineering effort.
As a more feasible but slightly less rigorous benchmark, we \emph{compiled} C code to x86 binary code, and then \emph{decompiled} the result using Ghidra.\footref{footnote:Ghidra}
This transformation can be implemented without modifying existing codebases, and should still give insight into decompiler correctness by letting us compare the decompiled source to the original.
% We manually inspected negative outcomes to rule out compiler bugs, but did not rule out the (theoretical) possibility of ``two wrongs making a right'', i.e., a bug in the decompiler unwittingly compensating a bug in the compiler, leading to a false positive outcome.

However, we immediately face the challenge of pairing the primitive actions from the decompiled code with actions in the source code.
For instance, a primitive action \lstinline{i += 1} in the source code could be decompiled to \lstinline{i++}.
Detecting \emph{all} such transformations would be outside the scope of our project.
To address this, we make the compiler \emph{blind} to the nature of the primitive actions and primitive tests by replacing them with calls to new functions \lstinline{void pact(int)} and \lstinline{bool pbool(int)}, respectively. The parameter distinguishes the primitives, and the correspondence in the decompiled code can be inferred from it. % chktex 36
Figure~\ref{c:blinded} shows the (abbreviated) \emph{blinded} version of our case study function.
This transformation does not alter control flow, but the blinded code can be longer, as the blinder also expands preprocessor macros.
Crucially, blinding depends on indicator variable detection, since indicator tests and assignments need to remain in the blinded code.

In this experiment, we use Ghidra as our decompiler, and \command{clang} as our compiler.
The C code obtained from compiling and then decompiling the blinded code (Figure~\ref{c:blinded})
and manually removing decompilation artifacts (see remark below) is shown in Figure~\ref{c:ghidra}.
The code produced is markedly different from the source --- it has three more \command{goto}s with as many labels --- yet our implementation is able to validate, in a fraction of a second,\footnote{\label{f:experiments}The experiments ran on a system with an Intel Core i7-9750H CPU @ 2.60GHz and 16 GB of RAM.} that this code is equivalent to the original.
\begin{remark}%
The following manual modifications were needed.
\begin{inparaenum}[(a)]
\item Ghidra assigns the return value of \lstinline{pbool} to a variable before reading it. This would be interpreted as a primitive action by our conversion. We manually propagate the assignments to avoid this.
\item In conditional expressions, Ghidra masks values with \lstinline{1} and then compares the result to \lstinline{0}, e.g.:\ \lstinline{(pbool(n)&1)!=0}. Again, our conversion would interpret this as a primitive action and not a test. We simplify this to \lstinline{pbool(n)}. % chktex 36
\end{inparaenum}
We believe that these steps could be automated, or Ghidra could be adapted to avoid them, but opted to perform them manually to keep our case study simple.
\end{remark}

\paragraph{Goto-elimination}
In general, \command{goto} statements can be eliminated by introducing additional indicator variables to guide control flow~\cite{yakdan_NoMoreGotos_2015,DBLP:journals/cacm/BohmJ66,erosa-hendren-1994}.
% TODO: We need an example here.
% CZ: is the reference enough?
This idea underpins Erosa and Hendren's \command{goto}-elimination algorithm~\cite{erosa-hendren-1994}.
Calipso~\cite{DBLP:journals/tsi/CasseFRS02} provides an improved implementation\footnote{\url{https://github.com/BinaryAnalysisPlatform/FrontC}} of their algorithm.

We ran Calipso on the blinded code in Figure~\ref{c:blinded}, and manually adjusted the output (see remark below).
The result again fundamentally different from the blinded source, as it contains a new indicator variable that is not present in the input, and does not contain any \command{goto}s.
Our tool confirms, again in a fraction of a second, that the code thus obtained is equivalent to the blinded input.
Indicator detection is crucial: if the newly-introduced variable \lstinline{factor_found} is not detected as an indicator, its assignments and tests will be converted into primitive actions and tests respectively.

\begin{remark}
The following manual modifications were needed for this experiment.
\begin{inparaenum}[(a)]
\item
We changed the type of the newly introduced variable \lstinline{factor_found} from \lstinline{char} to \lstinline{int}.
\item
We changed instances where \lstinline{factor_found} is used as a Boolean into comparisons (e.g., \lstinline{if (factor_found == 0)}).
\end{inparaenum}
As before, we consider these adaptations to be relatively minor; a future \command{goto}-elimination algorithm could prevent them from being necessary, or they could be performed automatically.
\end{remark}

\paragraph{Blinder coverage}
Overall, we consider our blinder to be quite promising: despite the fact that our current prototype does not support constructs like \command{switch}, we are able to blind 807 of the 1425 functions in Coreutils.
We use the cyclomatic complexity number (CCN)~\cite{1702388} as reported by Lizard\footnote{\url{http://www.lizard.ws/}} to provide an overview of how intricate these functions are.
The majority of the blinded functions have CCN up to 4, but we are able to blind some rather complex functions with the maximum CCN being 44.
Figure~\ref{t:barplot} provides a plot of the frequency of blinded functions per CCN\@.

CCN alone is not the discriminating factor that prevents us from blinding functions: a run of Lizard on the original source code of Coreutils shows that, while we are unable to blind functions with very high CCN (the maximum being 281), the median and mode CCN of functions we are not able to blind are 9 and 2, respectively, which shows there are relatively simple functions with respect to the CCN metric that we cannot handle.
In order to improve our coverage the focus should be on supporting more C constructs; for example, \command{switch} statements, \command{continue} statement, and ternary operators.
Such an extension requires some additions to our theory, as would support for other constructs.
Nevertheless, it is quite possible to embed special cases in the current theory.

% \section{Verifying Control Flow Analysis}

% One of the crucial steps in the process of decompilation is turning a control flow graph (CFG)
% into a structured imperative program.
% In this paper, the control flow graph is represented by a GKAT automaton
% and the program resulting from decompilation is modeled as a GKAT with Indicator expression.

% At the core of our method, we present a system to check
% for trace equivalence between a GKAT automaton and a GKATI expression,
% which give rises a nearly-linear time algorithm for checking such equivalence.
% This algorithm is \emph{complete} in the sense that
% the algorithm will output true if and only if the two input are trace equivalent.
% In the real world, users can treat the decompilation algorithm as a black box,
% and automatically evaluates the equivalence of input and output of the decompilation procedure.
% Compare to formally verifying the decompilation algorithm in a proof assistant,
% there are several advantages to the black-box approach:
% \begin{itemize}
%     \item The user does not need to understand our system or the decompilation procedure
%         to use our tool, since the decision procedure is automatic and agnostic
%         to the decompilation algorithm.
%     \item Our approach can verify each individual runs of the decompilation algorithm,
%         and does not require the decompilation algorithm to be correct on every input.
% \end{itemize}
% However, we also hypothesize that our framework can also be used as
% a efficient way to verify decompilation algorithm in general.
% TODO: example.

\begin{figure}
 \includegraphics[width=5.5cm]{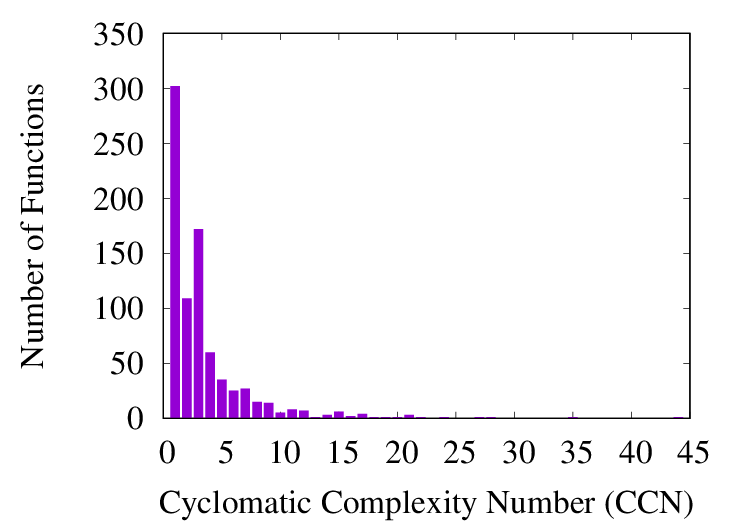}
 \Description{The plot of cyclomatic complexity number of blinded function, we can see that the number of function decreases as the CNN increase. Most of the blinded function have CCN smaller or equal to 5, but there are several functions with high CCN that we can blind.}
 \caption{Plot of the number of blinded functions per cyclomatic complexity number (CCN). The maximum CCN found was 44, yet the majority of the blinded functions have low CCN.}%
 \label{t:barplot}
\end{figure}

% !TEX root=../paper.tex

\section{Related Work}%
\label{sec:related-work}

% We identify and discuss related work in three distinct areas.
% Kleene algebra, program equivalence, and abstract interpretation.

\subsection{(De)compiler Verification and Validation} % chktex 36

Because we anticipate that the primary use case of CF-GKAT is the verification and validation of (de)compilers, we briefly go over existing work in this area, and relate it to ours. % chktex 36

\smallskip
\emph{Compiler verification and validation} is a well-developed research area.
The most rigorous approach is to use a proof assistant to verify the compilation algorithm relative to a formal semantics of the source and target language~\cite{leroy_CompCertFormallyVerified_2016,tan-etal-2016,neis-etal-2015}.
This method offers a foundational guarantee of correctness without imposing runtime overhead, but can be challenging and time-consuming to adapt on a large scale.
We envision that CF-GKAT could be used to automate correctness proofs of program transformations (provided that CF-GKAT equivalence is sound w.r.t.\ the language semantics).

Compiler testing feeds synthesized programs into different or differently configured compilers, and compares the resulting programs~\cite{chen_SurveyCompilerTesting_2021a,le_CompilerValidationEquivalence_2014a,yang_FindingUnderstandingBugs_2011}.
While this approach does not provide any formal guarantee of correctness, it has been effective in finding compiler bugs~\cite{yang_FindingUnderstandingBugs_2011}.
Within this framework, CF-GKAT could be used as an equivalence checking oracle; however, because our work is limited to \emph{control flow} equivalence, a blinding pass (similar to our experiments above) is likely necessary.

Translation validation~\cite{goldberg_LoopsPracticalIssues_2005,kasampalis_LanguageparametricCompilerValidation_2021,necula_TranslationValidationOptimizing_2000,kasampalis_TranslationValidationCompilation_2021,zhang_FormalVerificationOptimizing_2018,pnueli-etal-1998,sewell-etal-2013} scales well \emph{and} offers formal guarantees.
This is achieved through the use of symbolic evaluations and bisimulation to generate verification conditions for SMT solvers, allowing the validator to handle both data-flow and control-flow transformations simultaneously, albeit at the cost of decidability. Some approaches also employ graph isomorphism checks~\cite{awadhutkar-etal-2022}, which can be more computationally expensive than bisimulation.
Compared to other efforts, our focus is exclusively on control flow transformations.
This allows us to work with a minimal language, sufficient to cover well-known control-flow transformations~\cite{yakdan_NoMoreGotos_2015,erosa-hendren-1994,DBLP:journals/cacm/BohmJ66,grathwohl_KAT_2014}, and establish meta-properties such as correctness of our decision procedure and low complexity --- all of which are properties out of reach for unabstracted (Turing-complete) languages.

Predicate abstraction~\cite{ball-etal-2001} can help focus on relevant parts of a program prior to analysis, which makes it similar to blinding.
However, predicate abstraction requires input on \emph{how} to abstract the program, whereas blinding abstracts away all parts that are not relevant to flow control.

\smallskip
\emph{Decompiler verification and validation} is, by comparison, not as explored as compiler verifcation and validation.
Again, a distinction can be made between verification~\cite{DBLP:conf/pldi/VerbeekBFR22}, validation~\cite{dasgupta_ScalableValidationBinary_2020} and testing~\cite{liu-wang-2020,cao-etal-2024}.
Current efforts in verification and validation focus on \emph{lifting}, i.e., decoding machine code into assembly instructions or some other intermediate representation~\cite{ccs24verbeek,DBLP:conf/pldi/VerbeekBFR22,dasgupta_ScalableValidationBinary_2020}.
There has been some preliminary work on developing a sound decompiler that can produce C code, although control flow structuring is limited~\cite{DBLP:conf/sefm/VerbeekOR20}.
This constitutes a prime use case for CF-GKAT\@.

\subsection{Kleene Algebra and Control Flow Verification}

Existing work has explored non-local control-flow structures and indicator variables within the framework of KAT, albeit with a number of differences from our work.

Kozen characterized the semantics for programs with non-local control-flow structures as a family of KAT expressions~\cite{kozen_NonlocalFlowControl_2008a}.
This approach yields a decision procedure for program equivalences.
In contrast, CF-GKAT takes a more explicit approach by defining the continuation semantics, and the equivalence is computed by converting programs directly into automata.
This resolves an open question by \citet{kozen_NonlocalFlowControl_2008a}, on whether non-local control flow structures can be treated ``directly'', without being converted into KAT expressions.
The interested reader is referred to~\cite{tencate-kappe-2024} for an investigation of the limits of compositionally specified deterministic control flow.

\citet{grathwohl_KAT_2014} proposed \emph{KAT+B!} an extension of KAT with ``mutable tests'', which can be thought of as indicator variables with \(I = \{\true, \false\}\).
Concretely, their setter \(b!\) equates to indicator assignment \(x = \true\) and \(\overline{b}!\) to \(x = \false\); similarly, their tester \(b?\) corresponds to primitive indicator tests. % chktex 40
Although this is a special case of indicator variables, KAT+B!\@ can simulate indicator variable over a finite set \(I\) with \(|I|\) mutable tests.
% For example, indicator assignments $x := i$ can be simulated by ${b_i!} ;  \Pi _{i'  \neq  i}(\overline{b_{i'}}!)$, where each mutable test \(b_i\) records whether the indicator variable \(x\) is assigned to \(i\).
We opt to treat indicator assignments and tests as primitives, rather than restricting ourselves to (boolean-valued) mutable tests.

Our treatment of indicator variables also draws inspiration from NetKAT~\cite{anderson_NetKATSemanticFoundations_2014a}.
Specifically, NetKAT can be seen as a special case of KAT with indicator variables, the only primitive action is \texttt{dup}.
KATch~\cite{moeller-etal-2024}, a fast symbolic equivalence checker for NetKAT, is of particular interest to us in this space.
In the future, we hope to scale our equivalence checker using similar symbolic techniques.

To reason about variable assignment beyond indicator variables, Schematic KAT~\cite{angus_KleeneAlgebraTests_2001} provides a fine-grained algebraic theory for assignments over uninterpreted functions.
Later, Schematic KAT was also extended to reason about local variables~\cite{aboul-hosn_LocalVariableScoping_2008a}.
Neither work covers the complexity of the equivalence problem for schematic KAT and its extensions.
Kleene algebra can also be extended with nominal techniques~\cite{kozen_CompletenessIncompletenessNominal_2015,kozen_NominalKleeneCoalgebra_2015,gabbay_FreshnessNameRestrictionSets_2011}, which may help to reason about potentially infinite data domains, although the inclusion of tests to nominal Kleene algebra has not yet been investigated.

Separately, \citet{kozen_CertificationCompilerOptimizations_2000b} have applied KAT to verify compiler correctness.
Their system directly uses postulated equalities for parts of their verification task.
In contrast, our framework is based on the trace semantics, a commonly accepted semantics for \command{while}-programs.

\subsection{Complexity and Expressivity}

Finally, unlike the systems above, our system is based on GKAT, instead of KAT\@.
This enhances the scalability of our equivalence checking algorithms to accommodate larger programs: whereas equivalence of KAT expressions is PSPACE-complete~\cite{cohen-kozen-smith-1996}, equivalence of CF-GKAT expressions can be verified in polynomial time for a fixed number of tests~\cite{smolka_GuardedKleeneAlgebra_2020}.
Symbolic techniques previously applied to (Net)KAT may also provide better ways of mitigating the complexity of tests~\cite{pous-2015,moeller-etal-2024}. % chktex 36

Our system is the first to integrate both indicator variables and non-local control flow in a unified framework, which enables the verification of complex control-flow transformations that leverage both indicator variables and non-local control flow~\cite{yakdan_NoMoreGotos_2015}.
Our notion of trace equivalence is also coarser than previous systems; CF-GKAT equates programs ending in different indicator values:
\[
    \begin{array}{c}
    x := \true; \comITE{x = \true}{\command{print(1)}}{\command{print(2)}} \\[1mm]
    x := \false; \comITE{x = \false}{\command{print(1)}}{\command{print(2)}}
    \end{array}
\]
The two programs above are equivalent to an outside observer, as both of them will print 1; yet the assignment of \(x\) is different at the end of the so, thus systems like KAT+B!~\cite{grathwohl_KAT_2014} will not equate them.
On the other hand, our equivalence is not a congruence.
For example, concatenating $\comAssert{(x = true)}$ to the two equivalent programs above will yield inequivalent programs.

\section{Conclusion and Future Work}%
\label{sec:conclusion}

In this paper, we introduced CF-GKAT (Control-Flow GKAT), a system that extends GKAT with non-local control flow and indicator variables to validate control-flow transformation programs.
We formalized two semantics for CF-GKAT\@.
The first is the continuation semantics, where each trace, represented as a guarded word, is augmented with a continuation. The second is the trace semantics, which is obtained by resolving the continuations in the continuation semantics.

We proposed CF-GKAT automata as the operational model for CF-GKAT programs, and showed how to obtain these automata from CF-GKAT expressions while preserving their continuation semantics.
By lowering these automata into GKAT automata, we can then build an equivalence checker for CF-GKAT programs that scales nearly linearly with respect to the size of the input CF-GKAT programs.
Thus, our work provides an efficient validation algorithm for various control-flow transformations that utilize indicator variables and non-local control flow~\cite{yakdan_NoMoreGotos_2015,erosa-hendren-1994}.

While we successfully addressed one of Kozen's questions~\cite{kozen_NonlocalFlowControl_2008a} by presenting an algorithm to directly convert CF-GKAT programs into automata, we have yet to develop a coalgebraic perspective on non-local control flow utilizing Brzozowski derivatives~\cite{brzozowski_DerivativesRegularExpressions_1964}.
Such an approach could streamline several proofs, such as trace preservation of the lowering (\Cref{the:cf-gkat-automaton-lowering-correctness}) and the correctness of the operational semantics (\Cref{the:thompson-correctness}), and lead to a memory-efficient on-the-fly algorithm for trace equivalences between CF-GKAT programs.
A coalgebraic checker could also make use of symbolic techniques~\cite{pous-2015} to prevent explicit calculations based on the atoms of a Boolean algebra.

Additional future work could be the inclusion of the \command{continue} command within loops, as well as other types of control flow found in modern programming languages such as \command{do}-\command{while} and \command{switch}.
In terms of the theory's extensibility, it would be beneficial to separate the treatment of indicator variables and non-local control.
Currently, both components are integrated into the CF-GKAT automata signature as a unified entity.
While this approach provides a compact definition of operational semantics, it also introduces complexities when incorporating other non-local controls like \command{continue}.
Specifically, we will need to pass the indicator value in the continuation for \command{continue}, despite none of the non-local control-flow structures changing the indicator variable.

\begin{acks}
C.~Zhang was supported by \grantsponsor{chengnsf}{the US National Science Foundation}{https://dx.doi.org/10.13039/100000001} under awards no.~\grantnum{chengnsf}{CNS 1845803} and~\grantnum{chengnsf}{CNS 2040249}.
T.~Kappé was partially supported by \grantsponsor{tobiasmc}{the European Union’s Horizon 2020 research and innovation programme}{https://www.eeas.europa.eu/eeas/horizon-2020_en} under grant no.~\grantnum{tobiasmc}{101027412} (VERLAN), and partially by \grantsponsor{tobiasveni}{the Dutch research council (NWO)}{https://www.nwo.nl/} under grant no.~\grantnum{tobiasveni}{VI.Veni.232.286} (ChEOpS).
N.~Naus and D.~E.~Narváez are supported by \grantsponsor{nicodaviddarpa}{the Defense Advanced Research Projects Agency (DARPA)}{https://dx.doi.org/10.13039/100006502} and \grantsponsor{nicodavidniwc}{the Naval Information Warfare Center Pacific (NIWC Pacific)}{https://dx.doi.org/10.13039/100016744} under contract no.~\grantnum{nicodaviddarpa}{N66001-21-C-4028}.
Additionally, N.~Naus is also supported by by \grantsponsor{niconsf}{the US National Science Foundation}{https://dx.doi.org/10.13039/100000001} under award no.~\grantnum{niconsf}{CNS 2234257}.

We would like to thank Todd Schmid for the insightful discussions while preparing this work and the anonymous POPL reviewers for their thoughtful comments.
\end{acks}

\bibliographystyle{ACM-Reference-Format}
\bibliography{ref}

\end{document}